\documentclass[11pt]{article}
	
	\newcommand{\blind}{0}
	
    \makeatletter
    \renewcommand\section{\@startsection {section}{1}{\z@}%
        {-3.5ex \@plus -1ex \@minus -.2ex}%
        {2.3ex \@plus.2ex}%
        {\normalfont\fontfamily{phv}\fontsize{16}{19}\bfseries}}
    \renewcommand\subsection{\@startsection{subsection}{2}{\z@}%
        {-3.25ex\@plus -1ex \@minus -.2ex}%
        {1.5ex \@plus .2ex}%
        {\normalfont\fontfamily{phv}\fontsize{14}{17}\bfseries}}
    \renewcommand\subsubsection{\@startsection{subsubsection}{3}{\z@}%
        {-3.25ex\@plus -1ex \@minus -.2ex}%
         {1.5ex \@plus .2ex}%
         {\normalfont\normalsize\fontfamily{phv}\fontsize{14}{17}\selectfont}}
    \makeatother
	
	\usepackage{amsmath}
    \usepackage{bm}
	\usepackage{graphicx}
    \usepackage{subcaption}
	\usepackage{enumerate}
	\usepackage{xcolor}
	\usepackage{natbib} 
	\usepackage{url} 
	
    \usepackage{amsfonts, amsthm, latexsym, amssymb}
    \usepackage{lineno}
    \usepackage{cleveref}
    \usepackage{booktabs}
    \usepackage{siunitx}
    \usepackage{placeins}
    \usepackage{multirow}
    \usepackage{pbox}

    \newcommand{\titletext}{Multi-output Orthogonal Gaussian Processes for Noisy Simulator Outputs}
    \newcommand{\vect}[1]{\boldsymbol{#1}}
    \newcommand{\mat}[1]{\mathbf{#1}}
    \newcommand{\R}{\mathbb{R}}
    
    \newcommand{\E}{\mathbb{E}}
    \newcommand{\cov}[1]{\operatorname{Cov}({#1})}
    \renewcommand{\vec}[1]{\operatorname{vec}({#1})}
    \newcommand{\diag}[1]{\operatorname{diag}({#1})}
    \newcommand{\bdiag}[1]{\operatorname{bdiag}({#1})}
    \newcommand{\transpose}{^{\top}}
    \newcommand{\inv}{^{-1}}
    \newcommand{\eps}{\varepsilon}
    \newcommand{\tr}[1]{\operatorname{tr}\left(#1\right)}

    \newcommand{\Sigerror}{\mat{\Sigma}_\eps}
    \newcommand{\gpparam}{\vect{\theta}}
    
    \newtheorem{definition}{Definition}
    \newtheorem{theorem}{Theorem}
    \newtheorem{proposition}{Proposition}
    \newtheorem{lemma}{Lemma}

\begin{document}
		
		\def\spacingset#1{\renewcommand{\baselinestretch}%
			{#1}\small\normalsize} \spacingset{1}
		
		\if0\blind
		{
			\title{\bf {\titletext}}
            \author{Evan C. Barnett$^a$ and Moses Y.-H. Chan$^b$ \\ $^a$ Department of Engineering Sciences and Applied Mathematics \\ $^b$ Department of Industrial Engineering and Management Sciences\\Northwestern University, Evanston, Illinois, USA}
            \date{July 31, 2026}
			\maketitle
		} \fi
		
		\if1\blind
		{

            \title{\bf {\titletext}}
			\author{Author information is purposely removed for double-masked review}
			
\bigskip
			\bigskip
			\bigskip
			\begin{center}
				{\LARGE\bf {\titletext}}
			\end{center}
			\medskip
		} \fi
		\bigskip
		
\begin{abstract}

Computer simulations can model physical processes but are often too expensive to produce enough runs for calibration, sensitivity analysis, prediction, and uncertainty quantification. As a result, statistical surrogates are frequently used as cheaper alternatives that can be trained on a small number of simulation runs. Gaussian processes (GP) are well-suited as surrogates, as they provide flexible, nonlinear regression and closed-form uncertainty quantification. However, standard GPs are insufficient replacements for more complex simulators, such as stochastic or multi-output simulators. In addition, typical GPs also struggle to separate fitted regression coefficients from residual process variation. We introduce multi-output orthogonal Gaussian process (MOOGP) to tackle the problems mentioned above. This paper has three main contributions: (i) a definition of orthogonality for multi-output Gaussian processes, (ii) a construction that enforces orthogonality by conditioning each covariance function, and (iii) a structured likelihood formulation that significantly improves computational scaling. In a trend-recovery illustrative example, MOOGP recovers the true trend while the non-orthogonalized counterpart fails to do so, to the extent it reverses the sign of the trend. Numerical experiments and an application in heavy-ion collision simulations further demonstrate the interpretability and predictive advantage of MOOGP.  

\vspace{10pt}

	\end{abstract}
			
	\noindent%

	\spacingset{1.5} 

\section{Introduction} \label{s:intro}

Computer simulations are essential tools for studying the dynamics of complex systems when direct experimentation is costly, slow, or only partially observable. A simulator models the relationship between controllable inputs and outputs, which can be used to understand underlying behavior of the system and support design decisions. However, computer simulations are often computationally demanding to accurately represent the system of interest, for examples, in engineering design, nuclear physics, and climate modeling \citep{higdon2015bayesian,phillips2021band,plumlee2021high,chan2024constructing,li2026additive}. One setting in which these computational challenges arise is the Bayesian calibration and uncertainty quantification for a heavy-ion collision model \citep{phillips2021band,liyanage2023bayesian}, our primary application considered in Section~\ref{s:application}. In an application like this, only a limited number of simulator runs is feasible. This general statistical problem is well-researched in the computer experiments literature: A finite collection of simulator runs is selected, commonly using a space-filling design, and a statistical surrogate model is built to predict the simulator's output (with uncertainty estimates) at untried inputs \citep{sacks1989design,santner2003design,gramacy2020surrogates,joseph2026experimental}. Such surrogates enable prediction, sensitivity analysis, optimization, or calibration when these tasks cannot be performed directly through physical experimentation or computer simulation, for example, see \citet{kennedy2001bayesian}.

Gaussian process (GP) models are widely used as statistical surrogates, also called emulators, for computer experiments because they provide flexible nonparametric (near) interpolation or smoothing and closed-form uncertainty quantification \citep{rasmussen2006gaussian,gramacy2020surrogates}. A standard GP combines a mean function (e.g., regression) with a stochastic residual process. The mean function specifies the expected response, while the GP residual models the remaining variation. Modeling separate mean and residual components is attractive in engineering design and scientific discovery because the regression terms in the mean function can encode low-order trends and provide interpretation, while the GP component preserves flexibility. However, it is not guaranteed that the mean and residual are learning mutually exclusive components, a failure illustrated in \Cref{fig:low_order_trend}, where the fitted trend of a conventional GP misses a known linear effect badly enough to reverse its sign. Because the GP residual is not constrained in any way with respect to the mean space by default, it can absorb variation lying in the span of the regression basis, which leaves the regression coefficients unidentifiable and their interpretation unreliable. Orthogonal GP models address this issue by constraining the GP residual to lie within the orthogonal complement of the regression space, so that the residual cannot absorb variation intended for the mean function \citep{plumlee2018orthogonal}.

The present work is motivated by a setting in which this identifiability problem
arises together with two additional complications: noisy simulator evaluations
and multidimensional outputs. In many modern simulation studies, the simulator
output is not a single deterministic scalar.  A run may return a time series, a
spatial field, or several coupled physical quantities.
Moreover, simulator evaluations are frequently stochastic, so repeated runs at the
same input do not necessarily produce the same outputs
\citep{baker2022analyzing}. The randomness enters through two mechanisms. It may
be intrinsic to the code, generated by pseudo-random draws, as in the agent-based epidemic models  \citet{fadikar2018calibrating} or in Monte Carlo solvers that
approximate the solution of a deterministic differential equation by averaging
random sample paths \citep{herbei2014estimating}. Alternatively, it may be
injected by randomizing inputs that the experimenter does not control, so that the
reported output is a sample average over a distribution of those inputs
\citep{marrel2012global}. In either case, uncertainty in
the surrogate has at least two components: uncertainty about the mean function and
irreducible simulation noise, whose magnitude can vary substantially across output
components. Stochastic kriging and nugget-based
GP models provide tools for scalar stochastic simulation
\citep{ankenman2010stochastic,gramacy2012cases,binois2018practical}, but they do
not by themselves address cross-output dependence or ensure that regression coefficients are identifiable and interpretable.

When a computer simulation produces multiple outputs for each input setting, the emulator must predict a response vector rather than a single scalar \citep{fricker2013multivariate}. The simplest strategy to accomplish this is to fit independent scalar emulators to each output coordinate \citep{gu2016parallel,ogara2025hetgpy,surer2025puq}, but this discards dependence among outputs and can misrepresent joint predictive uncertainty. Separable covariance models impose a
Kronecker product between an input covariance and an output covariance, which
produces substantial computational benefits but forces all output dimensions to
share the same input-space correlation structure, e.g., see \citep{conti2010bayesian,hung2015analysis}.  More
flexible approaches use basis or latent-factor representations, in which the
high-dimensional output is expressed as a linear combination of lower-dimensional
latent processes. Representative constructions include principal-component emulators for
high-dimensional computer model output \citep{higdon2008computer,liyanage2023bayesian,chan2024constructing}, outer-product
and Kronecker emulators \citep{rougier2008efficient}, latent-factor and
multi-task Gaussian processes \citep{teh2005semiparametric,bonilla2008multi},
linear model of coregionalization and convolution-process constructions
\citep{alvarez2011computationally}, and more recent scalable
multi-output GP formulations \citep{bruinsma2020scalable}. \citet{alvarez2012kernels}, \citet{liu2018remarks}, and \citet{lartaud2022multi} give broader reviews of multi-output covariance
construction. These methods share
the central idea that output dependence can be learned through a structured
covariance or latent representation, rather than an unrestricted covariance over all
input-output pairs.

This paper introduces the multi-output orthogonal Gaussian process (MOOGP), a
surrogate model for noisy vector-valued simulator outputs. MOOGP represents the latent simulator response as the sum of a deterministic regression component and a multi-output GP residual. The residual is built from independent latent GPs mixed through an output basis, allowing the model to share information across output dimensions. The key difference from a standard multi-output GP is that the latent residual processes are orthogonalized with respect to the regression space. This prevents the residual process from absorbing variation intended for the mean function, while retaining the dependence-sharing advantages of latent multi-output GP models.

The main contributions of this work are threefold. First, we define a multi-output orthogonality condition that extends scalar orthogonal GP emulation to vector-valued stochastic processes. Second, we show how to enforce this condition within a multi-output GP by orthogonalizing the latent covariance functions. Third, we develop efficient computational methods that exploit the low-rank output basis and covariance structure, making likelihood evaluation scalable. Separate numerical examples show the effect of orthogonality on trend identification from the effect of output sharing on prediction. A controlled illustration example shows improved trend recovery under orthogonality, while a $2\times2$ ablation attributes the main predictive gain to output sharing. Subsequent benchmarks show that predictive differences between MOOGP and MOGP depend on the application.
\if0\blind
{
The proposed method is implemented in a PyPI-installable open-source package
\texttt{mooGP}. The functionalities to reproduce the numerical
results and figures in this paper are available at
\url{https://github.com/evancbarnett/mooGP}.
} \fi
\if1\blind
{
The proposed method is implemented in a PyPI-installable open-source package. The code and scripts to reproduce the numerical
results and figures in this paper are available in the submitted code folder for review.
} \fi

The remainder of the paper is organized as follows.  \Cref{s:background} reviews
univariate, orthogonal, and multi-output GP used for emulation.  \Cref{s:setting} introduces
the setting and notations, \Cref{s:moogp_model} sets up the MOOGP model, and \Cref{s:orthogonality} defines the multi-output orthogonality condition and effective construction of the MOOGP covariance.  \Cref{s:illustration} provides the controlled illustration examples in trend recovery with orthogonality and analysis of improvement sources. \Cref{s:numerical} details the numerical experiments to compare MOOGP with modern surrogate methods in a multi-output borehole example and the heavy-ion collision application.  \Cref{s:conclusions_discussion} closes the paper with remarks.

\section{Gaussian Processes Background} \label{s:background}

\subsection{Univariate Gaussian process} \label{s:background.1}

Let \(\mathcal{X}\subset \R^d\) denote the input domain and suppose first that a
simulator returns a scalar response.  Given training inputs
\(\vect{x}_1,\ldots,\vect{x}_n\), a standard GP emulator can be written as
\begin{align}
    y(\vect{x}) &= f(\vect{x})+\eps, \nonumber \\
    f(\vect{x}) &= \vect{g}(\vect{x})^\top \vect{\beta}+z(\vect{x}), \label{eq:univ_gp}
\end{align}
where \(\vect{g}(\vect{x})\in \R^r\) is a vector of fixed regression functions,
\(\vect{\beta}\in \R^r\) is a vector of regression coefficients, and
\(z(\cdot)\sim \mathcal{GP}(0,c(\cdot,\cdot;\gpparam))\) is a zero-mean
Gaussian process.  The error term \(\eps\) represents observation error,
simulation noise, or a nugget effect.  For deterministic computer experiments,
one often sets this term to zero or treats the nugget as a small numerical
regularization.  For stochastic simulations, \(\eps\) represents intrinsic
uncertainty and should be modeled as part of the predictive distribution
\citep{ankenman2010stochastic,gramacy2012cases}.

Let \(\vect{y}=(y(\vect{x}_1),\ldots,y(\vect{x}_n))^\top\),
\(\mat{G}=(\vect{g}(\vect{x}_1),\ldots,\vect{g}(\vect{x}_n))^\top\), and
\(\mat{C}= (c(\vect{x}_i,\vect{x}_j;\gpparam))_{i,j=1}^n\). Assume that training errors are independent of the latent process and satisfy
\( \eps_i\stackrel{\mathrm{iid}}{\sim}\mathcal{N}(0,\tau^2)\) for \( i=1,\ldots,n. \)
Then
\[
    \vect{y}\mid \vect{\beta},\gpparam,\tau^2
    \sim
    \mathcal{N}(\mat{G}\vect{\beta}, \mat{C}+\tau^2\mat{I}_n).
\]
Evaluating the resulting likelihood requires factorizing
\(\mat{C}+\tau^2\mat{I}_n\), which costs \(\mathcal{O}(n^3)\) and is the standard
bottleneck in exact GP inference.
For a new input \(\vect{x}'\), prediction follows from the conditional
distribution of a multivariate normal vector.  Writing
\(\vect{c}(\vect{x}')=(c(\vect{x}',\vect{x}_1;\gpparam),\ldots,c(\vect{x}',\vect{x}_n;\gpparam))^\top\),
the predictive mean is the fitted regression trend plus a covariance-weighted
residual correction,
\[
    \E[y(\vect{x}')\mid \vect{y}]
    =
    \vect{g}(\vect{x}')^\top \vect{\beta}
    +
    \vect{c}(\vect{x}')^\top(\mat{C}+\tau^2\mat{I}_n)^{-1}
    (\vect{y}-\mat{G}\vect{\beta}),
\]
with a corresponding predictive variance given by the usual Schur-complement
formula.  This structure explains the appeal of GP emulation: the same covariance
model governs both smoothing and uncertainty quantification.  It also reveals
the source of mean--residual confounding.  Both \(\vect{g}(\cdot)^\top
\vect{\beta}\) and \(z(\cdot)\) contribute to the fitted function, and without
additional restrictions, the GP residual is free to contain components that lie in
the span of the regression basis.

\subsection{Orthogonal Gaussian process} \label{s:background.2}

Orthogonal GP \citep{plumlee2018orthogonal} modifies the residual process in \Cref{eq:univ_gp} so
that it cannot explain variation in the regression space
\(\mathcal{G}=\operatorname{span}\{g_1,\ldots,g_r\}\).  For scalar functions
\(u,v\in L^2(\mathcal{X})\), define the inner product
\[
    \langle u,v\rangle
    =
    \int_{\mathcal{X}} u(\vect{x})v(\vect{x})\,d\vect{x}.
\]
Let \(\mathcal{G}^{\perp}=\{u\in L^2(\mathcal{X}):\langle u,\phi\rangle=0
\text{ for all }\phi\in \mathcal{G}\}\) denote the orthogonal complement of
\(\mathcal{G}\) in \(L^2(\mathcal{X})\).  The residual process is restricted to
\(\mathcal{G}^{\perp}\), so that
\(\langle g_j,z\rangle=0\) for each \(j=1,\ldots,r\) with probability one.  Under
mild regularity conditions, this can be achieved by replacing the original
kernel \(c(\cdot,\cdot;\gpparam)\) with the conditioned kernel
\citep{plumlee2018orthogonal}
\begin{equation} \label{eq:scalar_ortho_kernel_background}
    c^*(\vect{x},\vect{x}';\gpparam)
    =
    c(\vect{x},\vect{x}';\gpparam)
    -
    \vect{h}(\vect{x})^\top
    \mat{H}^{-1}
    \vect{h}(\vect{x}'),
\end{equation}
where
\[
    \vect{h}(\vect{x})
    =
    \int_{\mathcal{X}} c(\vect{x},\vect{\xi};\gpparam)
    \vect{g}(\vect{\xi})\,d\vect{\xi},
    \qquad
    \mat{H}
    =
    \int_{\mathcal{X}}\int_{\mathcal{X}}
    \vect{g}(\vect{\xi})
    c(\vect{\xi},\vect{\xi}';\gpparam)
    \vect{g}(\vect{\xi}')^\top
    d\vect{\xi}\,d\vect{\xi}'.
\]
The resulting process has covariance \(c^*\) and is orthogonal to
\(\mathcal{G}\) almost surely.  This construction preserves the flexibility of
the GP residual while assigning all variation in the regression space to the
mean component.  In the scalar-output setting, the result is a more identifiable
decomposition between systematic trend and stochastic residual.  MOOGP extends
this idea by imposing the same type of residual-space restriction on each
component of a vector-valued stochastic process.
\subsection{Multi-output Gaussian process} \label{s:background.3}

Now suppose the simulator has a latent response
\(\vect{f}(\vect{x})=(f_1(\vect{x}),\ldots,f_p(\vect{x}))^\top\in \R^p\).
A multi-output GP specifies a joint distribution over all outputs and
all input locations.  The modeling challenge is to capture dependence both
across the input space and across output dimensions without making inference
too expensive.  If each output is modeled independently, the joint predictive
distribution factorizes across outputs and its covariance is block diagonal.
This is computationally convenient but cannot represent cross-output dependence
or correlated predictive uncertainty.  At the other extreme, an unrestricted covariance over all $n$ inputs and $p$ output dimensions is highly flexible but difficult to parameterize, and
working directly with it requires factorizing a dense \(np\times np\) matrix at cost
\(\mathcal{O}(n^3p^3)\).

Structured covariance models sit between these extremes.  In a separable model, for any output dimensions $\ell, m \in\{1,\ldots,p\}$,
\[
    \cov{f_\ell(\vect{x}),f_m(\vect{x}')}
    =
    \Sigma_{\ell m}c(\vect{x},\vect{x}'),
\]
where \(c\) is a valid covariance kernel and
\(\mat{\Sigma}\in\R^{p\times p}\) is positive semidefinite and describes
output dependence.  For \(n\)
training inputs, let
\(\mat{F}=(\vect{f}(\vect{x}_1),\ldots,\vect{f}(\vect{x}_n))^\top\in
\R^{n\times p}\), and let \(\operatorname{vec}(\mat{F})\) stack its output
columns.  With \(\mat{C}=(c(\vect{x}_i,\vect{x}_{i'}))_{i,i'=1}^n\) and \( \cov{\operatorname{vec}(\mat{F})}=\mat{\Sigma}\otimes\mat{C}, \)
where
\(\otimes\) denotes the Kronecker product. This structure can be
exploited for efficient linear algebra \citep{rougier2008efficient}.  The main
limitation is that all outputs share the same input correlation function.  This
can be too restrictive when different output coordinates vary over different
input scales or represent different physical quantities
\citep{fricker2013multivariate}.

Basis-representation and latent-factor models relax this restriction. Let
\(\vect{z}(\vect{x})\in\mathbb{R}^q\) collect independent latent GPs, with
\[
    z_k(\cdot)\sim \mathcal{GP}(0,c_k(\cdot,\cdot;\gpparam_k)),
    \qquad k=1,\ldots,q,
\]
where every \(c_k\) is a valid covariance kernel.  Let
\(\mat{\Psi}=(\vect{\psi}_1,\ldots,\vect{\psi}_q)\in \R^{p\times q}\)
be an output loading matrix.  The residual process $\vect{w}(\vect{x})=\mat{\Psi}\vect{z}(\vect{x})$
has cross-covariance
\begin{equation} \label{eq:lmc_background}
    \cov{\vect{w}(\vect{x}),\vect{w}(\vect{x}')}
    =
    \sum_{k=1}^q
    \vect{\psi}_k\vect{\psi}_k^\top
    c_k(\vect{x},\vect{x}').
\end{equation}
This is the linear model of coregionalization form.  It includes separable
models as a special case but allows multiple latent correlation structures,
each contributing a rank-one output covariance.  When \(q<p\), the model also
acts as a dimension-reduction device, connecting it to principal-component and
latent-basis emulators for high-dimensional outputs
\citep{higdon2008computer,paulo2012calibration,higdon2015bayesian,liyanage2023bayesian,chan2024constructing}.

For noisy multi-output simulation, let
\(\mat{B}\in\mathbb{R}^{r\times p}\) denote the regression coefficient matrix.
An observation model can then be written as
\[
    \vect{y}(\vect{x}_i)
    =
    \mat{B}^\top \vect{g}(\vect{x}_i)
    +
    \mat{\Psi}\vect{z}(\vect{x}_i)
    +
    \vect{\eps}_i,
    \qquad
    \vect{\eps}_i\stackrel{\mathrm{iid}}{\sim}
    \mathcal{N}_p(\vect{0},\Sigerror),
    \qquad
    \Sigerror
    =
    \operatorname{diag}(\sigma_{\epsilon,1}^2,\ldots,
    \sigma_{\epsilon,p}^2),
\]
with every
\(\sigma_{\epsilon,\ell}^2>0\).  The implementation optionally partitions the
ordered outputs into declared contiguous groups and constrains the variances to
be equal within each group.  Singleton groups, which allow a separate variance
for each output, are the default.
Let
\(\mat{Y}=(\vect{y}(\vect{x}_1),\ldots,\vect{y}(\vect{x}_n))^\top\).
Under the stacking convention above, the covariance induced by
\Cref{eq:lmc_background} is
\[
    \mat{K}_{\mathrm{MOGP}}
    :=
    \cov{\operatorname{vec}(\mat{Y})}
    =
    \sum_{k=1}^q
    \vect{\psi}_k\vect{\psi}_k^\top \otimes \mat{C}_k
    +
    \Sigerror\otimes \mat{I}_n,
\]
where \(\mat{C}_k=(c_k(\vect{x}_i,\vect{x}_{i'};\gpparam_k))_{i,i'=1}^n\).  This is
the baseline MOGP structure used in the proposed model.  MOOGP keeps the same
latent output representation and observation-error model, but replaces each
\(\mat{C}_k\) by a matrix \(\mat{C}^*_k\) that is implied by an orthogonalized kernel
\(c_k^*\).  Consequently, the model shares information across output dimensions
through \(\mat{\Psi}\), accounts for simulation noise through \(\Sigerror\), and
enforces an interpretable decomposition between the regression component and the
multi-output residual process.

\section{Multi-output orthogonal GP} \label{s:main_model}

\subsection{Setting and notations} \label{s:setting}

We seek a surrogate for a stochastic simulator whose output at an input
$\vect{x}\in\mathcal{X}\subset\R^d$ is a noisy, vector response
$\vect{y}(\vect{x})\in\R^p$. Suppose $n$ simulator runs are available at inputs
$\vect{x}_1,\ldots,\vect{x}_n$. The goals are to predict the
latent response $\vect{f}(\vect{x})$ at untried inputs, identify the linear trend, and quantify simulation noise.
We retain the multi-output notation of \Cref{s:background.3}: regression
functions $\vect{g}(\cdot)\in\R^r$ with coefficient matrix $\mat{B}$, latent
processes $\vect{z}^*(\cdot)\in\R^q$ mixed through the loading matrix
$\mat{\Psi}$, and error covariance $\Sigerror$.
Throughout, \(i\) and \(i'\) index simulator runs, \(j\) indexes input coordinates,
\(s\) indexes regression functions, \(k\) indexes latent components, and \(\ell\)
indexes outputs.

The $n$ runs are stacked row-wise into the input, data, regression, latent, and error
matrices $\mat{X}\in\R^{n\times d}$, $\mat{Y}\in\R^{n\times p}$,
$\mat{G}\in\R^{n\times r}$,
$\mat{Z}^*\in\R^{n\times q}$, and $\mat{E}\in\R^{n\times p}$, whose $i$-th rows are
$\vect{x}_i\transpose$, $\vect{y}(\vect{x}_i)\transpose$,
$\vect{g}(\vect{x}_i)\transpose$,
$\vect{z}^*(\vect{x}_i)\transpose$, and $\vect{\eps}_i\transpose$,
respectively. Each output column is centered at its training mean and divided by
its training sample standard deviation, so $\mat{Y}$ denotes the standardized
response throughout. 

We write $\mat{I}_n$ for the $n\times n$ identity matrix, $|\mat{A}|$ for the determinant of $\mat{A}$, $\vec{\mat{A}}$ for the column-stacking vectorization of $\mat{A}$, and
$\bdiag{\cdot}$ and $\diag{\cdot}$ denoting the constructing of block-diagonal and diagonal
matrices.

Our results rely on the Woodbury identities for the matrix inverse and
determinant \citep{harville2008matrix}: For nonsingular $\mat{R}\in\R^{n\times n}$,
$\mat{T}\in\R^{p\times p}$, and matrices $\mat{S}\in\R^{n\times p}$,
$\mat{U}\in\R^{p\times n}$,
\begin{align}
    (\mat{R}+\mat{S}\mat{T}\mat{U})\inv
    &= \mat{R}\inv-\mat{R}\inv\mat{S}
    (\mat{T}\inv+\mat{U}\mat{R}\inv\mat{S})\inv
    \mat{U}\mat{R}\inv,
    \label{eq: woodbury}\\
    |\mat{R}+\mat{S}\mat{T}\mat{U}|
    &=|\mat{R}|\,|\mat{T}|\,|\mat{T}\inv+\mat{U}\mat{R}\inv\mat{S}|. \nonumber
\end{align}

\subsection{Main model} \label{s:moogp_model}

The orthogonalized kernel $c_k^*$ is constructed in
\Cref{s:orthogonality}. The MOOGP model is
\begin{align} \label{eq:mogp_model}
    \vect{y}(\vect{x}) &= \vect{f}(\vect{x})+\vect{\eps}, \\
    \vect{f}(\vect{x}) &= \mat{B}\transpose \vect{g}(\vect{x})
    + \mat{\Psi} \vect{z}^*(\vect{x}), \nonumber \\
    z_k^*(\cdot) &\sim
    \mathcal{GP}\!\left(0,c_k^*(\cdot,\cdot;\vect{\theta}_k)\right),
    \quad k=1,\ldots,q. \nonumber
\end{align}
Here, $\vect{x}\in\mathcal{X}$ is the input,
$\vect{y}(\vect{x})\in\R^p$ is the observed output at $\vect{x}$,
$\vect{f}(\vect{x})\in\R^p$ is the latent simulator response, and
$\vect{\eps}\in\R^p$ is its observation error. For fixed model parameters,
$z_1^*,\ldots,z_q^*$ are mutually independent. At the training inputs,
    \( \vect{\eps}_i\stackrel{\mathrm{iid}}{\sim}
    \mathcal{N}_p(\vect{0},\Sigerror), \) for \(i=1,\ldots,n, \) , where
$\Sigerror=\diag{\sigma_1^2,\ldots,\sigma_p^2}$ and
$\sigma_\ell^2>0$ for every output. Thus, the outputs are conditionally
independent given $\vect{f}(\vect{x})$.

The response contains a deterministic regression term
$\mat{B}\transpose\vect{g}(\vect{x})$ and a stochastic term
$\mat{\Psi}\vect{z}^*(\vect{x})$. MOOGP constrains these terms to be orthogonal
using the condition in \Cref{s:orthogonality}. The regression functions may
contain first-order effects and interactions, as in
\citet{plumlee2018orthogonal}. The remaining variation is modeled with
latent process vector $\vect{z}^*(\cdot)\in\R^q$ and loading matrix
 \( \mat{\Psi}=(\vect{\psi}_1,\ldots,\vect{\psi}_q)\in\R^{p\times q}, \) where
\( q\leq p. \)

\begin{proposition}[Log-likelihood]
    Under the model, i.e., \Cref{eq:mogp_model}, the data $\mat{Y}$ have a Gaussian log-likelihood 
    \begin{multline}
    \mathcal{L}(\mat{B}, \mat{\Theta}, \mat{\Psi}, \Sigerror \mid \mat{X}, \mat{Y})
    =\textrm{const}-\tfrac12 \log|\mat{K}_y|
    -\tfrac12 \big(\vec{\mat{Y}}- \vect{\mu}_y\big)^\top \mat{K}_y^{-1}\big(\vec{\mat{Y}}-\vect{\mu}_y\big), \\
    \vect{\mu}_y := (\mat{I}_p \otimes \mat{G})\vec{\mat{B}}, \quad 
    \mat{K}_y := (\mat{\Psi}\otimes\mat{I}_n)\mat{C}^*(\mat{\Psi}\transpose\otimes\mat{I}_n) + \Sigerror \otimes \mat{I}_n,
    \end{multline}
    where
    $\mat{C}^*=\bdiag{\mat{C}_1^*,\ldots,\mat{C}_q^*}$,
    $\mat{C}_k^*
    =(c_k^*(\vect{x}_i,\vect{x}_{i'};\gpparam_k))_{i,i'=1}^n$,
    and $\mat{\Theta}=(\gpparam_k)_{k=1}^q$ collects all GP hyperparameters.
\end{proposition}
\begin{proof}
The vectorized observation model is
\[
\vec{\mat{Y}}
=(\mat{I}_p\otimes\mat{G})\vec{\mat{B}}
+(\mat{\Psi}\otimes\mat{I}_n)\vec{\mat{Z}^*}+\vec{\mat{E}}.
\]
Its expectation and covariance are
\begin{align*}
    \E[\vec{\mat{Y}} \mid \mat{B}, \mat{\Theta}, \mat{\Psi}, \Sigerror] &= (\mat{I}_p \otimes \mat{G})\vec{\mat{B}}, \\
    \cov{\vec{\mat{Y}} \mid \mat{B}, \mat{\Theta}, \mat{\Psi}, \Sigerror} &= (\mat{\Psi}\otimes\mat{I}_n)\mat{C}^*(\mat{\Psi}\transpose\otimes\mat{I}_n) + \Sigerror \otimes \mat{I}_n.
\end{align*}
The Gaussian log-likelihood follows because $\mat{Y}$ is a sum of Gaussian variables.
\end{proof}
\begin{proposition}[Predictive equations]
    For any new input $\vect{x}'$, define
    \[
    \vect{c}_k^*(\vect{x}')
    :=
    \big(c_k^*(\vect{x}_i,\vect{x}';\gpparam_k)\big)_{i=1}^n,
    \qquad
    \mat{C}_{\vect{x}'}^*
    :=
    \bdiag{\vect{c}_1^*(\vect{x}'),\ldots,\vect{c}_q^*(\vect{x}')}
    \in\R^{nq\times q}.
    \]
    The predictive equations are
    \begin{align}
    \E[\vect{y}(\vect{x}') \mid \mat{X},\mat{Y} ; \mat{B}, \mat{\Theta}, \mat{\Psi}, \Sigerror] &= \mat{B} \transpose\vect{g}(\vect{x}')+\mat{K}_{y'y}\mat{K}_y\inv(\vec{\mat{Y}}-\vect{\mu}_y), \\
    \cov{{\vect{y}(\vect{x}')} \mid \mat{X},\mat{Y}; \mat{B}, \mat{\Theta}, \mat{\Psi}, \Sigerror} &= \mat{K}_{y'} - \mat{K}_{y'y} \mat{K}_{y}\inv \mat{K}_{yy'}, \\
    \mat{K}_{y'y} \transpose = \mat{K}_{yy'} 
    &:= (\mat{\Psi} \otimes \mat{I}_n) \mat{C}_{\vect{x}'}^* \mat{\Psi}\transpose, \nonumber \\
    \mat{K}_{y'}&:= \mat{\Psi}
    \diag{\big(c_k^*(\vect{x}',\vect{x}';\gpparam_k)\big)_{k=1}^q}
    \mat{\Psi}\transpose + \Sigerror. \nonumber
\end{align}
\end{proposition}

\begin{proof}
To calculate the predictive mean and variance, we first look at the joint distribution between the training observations $\vec{\mat{Y}}$ and the observation we would like to predict $\vect{y}(\vect{x}')$. 
\begin{equation*}
\begin{bmatrix}
    \vec{\mat{Y}} \\
    \vect{y}(\vect{x}')
\end{bmatrix}
\sim
\mathcal{N}
\left( 
\begin{bmatrix}
    (\mat{I}_p \otimes \mat{G})\vec{\mat{B}} \\
    \mat{B} \transpose \vect{g}(\vect{x}')
\end{bmatrix},
\begin{bmatrix}
    \mat{K}_y & \mat{K}_{yy'} \\
    \mat{K}_{y'y} & \mat{K}_{y'}
\end{bmatrix}
\right)
\end{equation*}
We obtain the result by applying the conditional distribution results from Gaussian distributions, e.g.,  \citet[ch. 5]{hardle2015multivariate}.
\end{proof}
\subsection{Orthogonality condition for multi-output processes} \label{s:orthogonality}
One of the main goals of the MOOGP is to simultaneously estimate the regression component $\mat{B} \transpose \vect{g}(\vect{x})$ and the stochastic residual process $\mat{\Psi} \vect{z}^*(\vect{x}).$  Without an explicit constraint, the Gaussian process could absorb variation from the regression space, leading to unidentifiability of the regression coefficients \citep{plumlee2018orthogonal}.  As a result, any reasonable interpretation of the regression coefficients requires some condition to separate $\mat{B}\transpose \vect{g}(\vect{x})$ and $\mat{\Psi} \vect{z}^*(\vect{x})$. 
We formalize this separation by requiring the stochastic residual process to lie in the orthogonal complement of the regression space over the entire input domain. This approach can be considered a natural extension of the method used in the scalar-output case \citep{plumlee2018orthogonal}.

\begin{definition}[Multi-output orthogonal Gaussian process] \label{def: MOOGP}
Fix the regression basis $\vect{g}(\cdot)$ and its span $\mathcal{G}$.
Given $\mat{B}\in\R^{r\times p}$ and a fixed, deterministic
$\mat{\Psi}\in\R^{p\times q}$, let
$z_1^*,\ldots,z_q^*$ be mutually independent, mean-zero Gaussian
processes, and set
$\vect{z}^*(\vect{x})
=(z_1^*(\vect{x}),\ldots,z_q^*(\vect{x}))\transpose$. Then
$\vect{w}(\vect{x}):=\mat{\Psi}\vect{z}^*(\vect{x})\in\R^p$ is a
jointly Gaussian vector process.
Let $w_\ell (\vect{x})$ denote the $\ell$-th component of $\vect{w}(\vect{x})$. Define the multi-output Gaussian process
\[
\vect{f}(\vect{x}):= \mat{B}\transpose \vect{g}(\vect{x})+\vect{w}(\vect{x}).
\]
We say that $\vect{f}(\cdot)$ is a multi-output orthogonal Gaussian process with regression space $\mathcal{G}$ if the stochastic process satisfies the orthogonality constraints
$\langle g_s, w_\ell \rangle = 0$ for all $s=1,\ldots,r$ and $\ell=1,\ldots,p$ with probability one. Equivalently in matrix form,
\begin{align} \label{eq:ortho_cond}
\int_\mathcal{X} \vect{g}(\vect{x})\vect{w}(\vect{x})\transpose \, d\vect{x}= \mat{0}_{r \times p}  \quad w.p.~1.
\end{align}
\end{definition}
To illustrate, take any function $\phi(\cdot) \in \mathcal{G}$. It can be written as
$\phi(\vect{x})=\sum_{s=1}^r a_sg_s(\vect{x})$ for some
$\vect{a}\in\R^r$. Then, for each $\ell$,
$\langle \phi, w_\ell \rangle
=\sum_{s=1}^r a_s\langle g_s,w_\ell\rangle=0$.
Thus, $w_\ell(\cdot)\in\mathcal{G}^\perp$ almost surely for each
$\ell=1,\ldots,p$.
The condition in \Cref{eq:ortho_cond} ensures that, for each output dimension, $\mat{B}\transpose \vect{g}(\cdot)$ captures the portion of the response in $\mathcal{G}$, while the Gaussian process captures only variation in $\mathcal{G}^\perp$.

For each latent component $k=1,\ldots,q$, let
$z_k\sim\mathcal{GP}(0,c_k)$ be an GP where $c_k(\cdot,\cdot;\gpparam_k)$ is a bounded,
continuous covariance function on $\mathcal{X}\times\mathcal{X}$. Define
\(
    \vect{A}_k
    :=
    \int_{\mathcal{X}}
    \vect{g}(\vect{\xi})z_k(\vect{\xi})\,d\vect{\xi}.
\)
Let
\begin{align} \label{eq: h_ints}
\vect{h}_k(\vect{x}) 
&:= \int_{\mathcal{X}} c_k(\vect{x},\vect{\xi};\gpparam_k)\,\vect{g}(\vect{\xi})\,d\vect{\xi}\in\mathbb{R}^r,  \nonumber   \\ 
\mat{H}_k 
&:= \int_{\mathcal{X}}\int_{\mathcal{X}}
\vect{g}(\vect{\xi})\,c_k(\vect{\xi},\vect{\xi}';\gpparam_k)\,\vect{g}(\vect{\xi}')^\top\,
d\vect{\xi}\,d\vect{\xi}'\in\mathbb{R}^{r\times r}.
\end{align}
Then
\[
    \operatorname{Cov}\{z_k(\vect{x}),\vect{A}_k\}
    =\vect{h}_k(\vect{x}),
    \qquad
    \operatorname{Cov}(\vect{A}_k)=\mat{H}_k.
\]
Assuming each $\mat{H}_k$ is finite and positive definite,
conditioning $z_k$ on $\vect{A}_k=\vect{0}$ results in a mean-zero GP
with covariance
\begin{equation} \label{eq:ck_star}
c_k^*(\vect{x},\vect{x}';\gpparam_k)
:= c_k(\vect{x},\vect{x}';\gpparam_k)
- \vect{h}_k(\vect{x})^\top \mat{H}_k^{-1}\vect{h}_k(\vect{x}').
\end{equation}
The latent process $z_k^*$ is defined to have this conditional law.
\begin{lemma}\label{lem: ortho_kernel_psd}
Assume $\mathcal{X}\subset\mathbb{R}^d$ is bounded, 
$c_k(\cdot,\cdot;\gpparam_k)$ is a bounded, continuous covariance function on 
$\mathcal{X}\times\mathcal{X}$, $g_s\in L^2(\mathcal{X})$ for
$s=1,\ldots,r$, and the matrix
$\mat{H}_k$ in \Cref{eq: h_ints} exists and is positive definite.  
Then, the kernel $c_k^*(\cdot,\cdot;\gpparam_k)$ defined in \Cref{eq:ck_star} is positive semidefinite on
$\mathcal{X}$ and hence a valid covariance function \citep{plumlee2018orthogonal}.
\end{lemma}
\begin{theorem} \label{thm: ortho_kernel}
Under the assumptions of \Cref{lem: ortho_kernel_psd} for
$k=1,\ldots,q$, suppose in addition that $\vect{g}(\cdot)$ is bounded
on $\mathcal{X}$. Let $z_1^*,\ldots,z_q^*$ be mutually independent,
mean-zero Gaussian processes, with $z_k^*$ having covariance
$c_k^*(\cdot,\cdot;\gpparam_k)$ defined in \Cref{eq:ck_star}. Let
$\mat{\Psi}$ be fixed and deterministic.
Then:
\begin{enumerate}
\item[(i)] 
\[
\int_{\mathcal{X}} \vect{g}(\vect{x})\,z_k^*(\vect{x})\,d\vect{x}=\vect{0}\in\mathbb{R}^r
\quad\text{with probability one}, \qquad k=1,\ldots,q.
\]
\item[(ii)] Let $\vect{z}^*(\vect{x})=(z_1^*(\vect{x}),\ldots,z_q^*(\vect{x}))^\top$ and define
$\vect{w}(\vect{x})=\mat{\Psi}\vect{z}^*(\vect{x})\in\mathbb{R}^p$. Then $\vect{w}(\cdot)$ satisfies
\[
\int_{\mathcal{X}} \vect{g}(\vect{x})\,\vect{w}(\vect{x})^\top\,d\vect{x}
=\mat{0}_{r\times p}
\quad\text{with probability one}.
\]
Therefore, $\vect{f}(\vect{x})=\mat{B}^\top\vect{g}(\vect{x})+\vect{w}(\vect{x})$ is a multi-output orthogonal Gaussian process, following \Cref{def: MOOGP}.
\end{enumerate}
\end{theorem}
\Cref{app:proof_ortho_kernel} gives the proof. The conditions for \Cref{thm: ortho_kernel} are sufficient. 
In the next section, we show how the integrals in \Cref{eq: h_ints} have an analytical form for specific choices of $c_k(\cdot,\cdot)$, and how the latent-structure induces a covariance with efficient Kronecker forms, enabling the fast computation of $\mat{K}_y$ and $c_k^*(\cdot,\cdot)$.   
\subsection{Efficient computation of the covariance matrix}
\label{s:methods.3}
The MOOGP covariance in the log-likelihood can be written as
\[
\mat{K}_y = \sum_{k=1}^q \vect{\psi}_k \vect{\psi}_k\transpose \otimes \mat{C}_k^*  + \Sigerror \otimes \mat{I}_n,
\]
where the orthogonalized blocks $\mat{C}_k^*$ enforce the multi-output
orthogonality condition in \Cref{eq:ortho_cond}.

The orthogonalized kernels \(c_k^*(\cdot,\cdot;\gpparam_k)\) from \Cref{thm: ortho_kernel} enforce the orthogonality constraint \Cref{eq:ortho_cond}, but their practical usefulness is limited if the method is computationally inefficient. In particular, each evaluation of the marginal log-likelihood requires repeated construction of the covariance blocks $\mat{C}_k^*$ and repeated construction of variants of the output covariance $\mat{K}_y\inv$ and $\log|\mat{K}_y|$. Because $\mat K_y \in \mathbb R^{np\times np}$, direct computation scales poorly. We therefore exploit the structure of $\mat{C}_k^*$ and $\mat{K}_y$ to improve computational efficiency.

\subsubsection{Orthogonalized matrix calculations}
For each latent component $k=1,\dots,q$, let $\mat{W}_k =(\vect{h}_k(\vect{x}_i)\transpose)_{i=1}^n \in \R^{n \times r}$, where $\vect{h}_k(\cdot)$ and $\mat{H}_k$ are defined in \Cref{eq: h_ints}. Then, the orthogonalized kernel $c_k^*(\cdot,\cdot)$ has the covariance matrix
$ \mat{C}_k^* = \mat{C}_k - \mat{W}_k \mat{H}_k\inv \mat{W}_k \transpose.$ For each evaluation of $\mat{C}_k^*$, we must compute the high-dimensional integrals $\vect{h}_k(\cdot)$ and $\mat{H}_k$. Without further simplification, evaluating these integrals would significantly hinder computation. We therefore follow the general procedure to convert $\vect{h}_k(\cdot)$ and $\mat{H}_k$ to a product of one-dimensional integrals, which is detailed in \cite{plumlee2018orthogonal}.

Assume $\mathcal{X}=\prod_{j=1}^d [a_j,b_j]$, the covariance is separable, and each regression basis function is a product of one-dimensional functions. Fubini's theorem then reduces the high-dimensional integrals $\vect{h}_k(\cdot)$ and $\mat{H}_k$ to products of one-dimensional integrals. The implementation maps each coordinate affinely to $[-1,1]$ and uses
\[
c_k(\vect{x},\vect{x}';\gpparam_k)
=\sigma_k^2\exp\left\{-\sum_{j=1}^d
\frac{(x_j-x_j')^2}{\ell_{kj}^2}\right\},
\qquad
g_s(\vect{x})=\prod_{j\in\mathcal{J}_s}x_j,
\]
where $\ell_{kj}>0$ and each $\mathcal{J}_s$ is a subset of $\{1,\ldots,d\}$. For this specification, only zeroth- and first-order one-dimensional Gaussian moments are required, and symmetry makes $\mat{H}_k$ diagonal. \Cref{app:analytic_integrals} provides the closed forms for $\vect{h}_k(\cdot)$, $\mat{H}_k$, and $c_k^*(\cdot,\cdot)$.

\subsubsection{Covariance matrix calculation}
Although we have an efficient way to calculate $\mat{C}_k^*$, direct inversion of $\mat{K}_y$ is still computationally prohibitive with time complexity $\mathcal{O}(n^3p^3)$. We first apply the Woodbury identities (\Cref{eq: woodbury}) in an attempt for further simplification.
\begin{align}
    \mat{K}_y\inv &= \Sigerror\inv \otimes \mat{I}_n - (\Sigerror\inv \mat{\Psi} \otimes \mat{I}_n)[\mat{\Psi}\transpose \Sigerror\inv \mat{\Psi} \otimes \mat{I}_n + {\mat{C}^*}\inv]\inv (\mat{\Psi}\transpose \Sigerror\inv \otimes \mat{I}_n), \label{eq: Ky_inv1} \\
    |\mat{K}_y| &= |\Sigerror \otimes \mat{I}_n| \; |\mat{C}^*| \; |\mat{\Psi}\transpose \Sigerror\inv \mat{\Psi} \otimes \mat{I}_n + {\mat{C}^*}\inv | \label{eq: Ky_inv2}
\end{align}
In \Cref{eq: Ky_inv1}, the first term requires inverting $\Sigerror$, which is a
$p\times p$ matrix. The other inverse involves
$\mat{\Psi}\transpose\Sigerror\inv\mat{\Psi}\otimes\mat{I}_n
+{\mat{C}^*}\inv$, an $nq\times nq$ matrix. This matrix remains too large for
practical surrogate use. The following structure diagonalizes the
$q\times q$ interaction term
$\mat{\Psi}\transpose\Sigerror\inv\mat{\Psi}$. Let
\begin{align} \label{eq: Psi_param}
    \mat{\Psi} &= \Sigerror^{\frac{1}{2}} \mat{\Phi}, \quad \mat{\Phi}\transpose \mat{\Phi} = \mat{D}
\end{align}
where $\Sigerror^{1/2}$ is the symmetric positive-definite square root of
$\Sigerror$, $\Sigerror^{-1/2}$ is its inverse,
$\mat{\Phi}=(\vect{\phi}_1,\ldots,\vect{\phi}_q)\in\R^{p\times q}$ has mutually
orthogonal columns, and $\mat{D}=\diag{d_1,\ldots,d_q}$. The square root exists
because $\Sigerror$ is symmetric positive definite
\citep[ch.~21]{harville2008matrix}, and is immediate under the diagonal error
model of \Cref{s:moogp_model}.

Let $\mat{Y}=\mat{U}\mat{S}\mat{V}\transpose$ be its
thin singular value decomposition, with singular values
$s_1\geq\cdots\geq s_q>0$ and leading right singular vectors
$\mat{V}_q=(\vect{v}_1,\ldots,\vect{v}_q)$. Setting
\[
\mat{\Phi}
=\sqrt{n}\,\mat{V}_q\diag{s_1^{-1},\ldots,s_q^{-1}},
\qquad
\mat{D}
=\diag{n s_1^{-2},\ldots,n s_q^{-2}}
\]
satisfies $\mat{\Phi}\transpose\mat{\Phi}=\mat{D}$ and is the choice used in the
reported fits. The empirical basis $\mat{\Phi}$ and $\mat{D}$ are
treated as fixed plug-in quantities during optimization, and their
estimation uncertainty is not propagated. The loading matrix
$\mat{\Psi}$ varies with estimated parameters only through
$\Sigerror^{1/2}$.

Under \Cref{eq: Psi_param},
$\mat{\Psi}\transpose\Sigerror\inv\mat{\Psi}=\mat{D}$, and the
$nq\times nq$ matrix in \Cref{eq: Ky_inv1} becomes block diagonal:
\[
\mat{D}\otimes\mat{I}_n+{\mat{C}^*}\inv
=\bdiag{
d_1\mat{I}_n+{\mat{C}_1^*}\inv,\ldots,
d_q\mat{I}_n+{\mat{C}_q^*}\inv}.
\]
Substituting this expression into \Cref{eq: Ky_inv1} yields
\begin{align}
    \mat{K}_y \inv &= \Sigerror\inv \otimes \mat{I}_n - \sum_{k=1}^q \Sigerror^{-\frac{1}{2}} \vect{\phi}_k \vect{\phi}_k \transpose \Sigerror^{-\frac{1}{2}} \otimes (d_k \mat{I}_n + {\mat{C}_k^*}\inv)\inv .
\end{align}
A similar simplification applies to \Cref{eq: Ky_inv2}, resulting in 
\[
|\mat C^*|\;\big|\mat D\otimes \mat I_n + {\mat C^*}^{-1}\big|
= \prod_{k=1}^q |\mat C_k^*|\;|d_k\mat I_n + {\mat C_k^{*}}\inv|
= \prod_{k=1}^q \big|\mat I_n + d_k \mat C_k^*\big|
\]
As a result of these reductions, each likelihood evaluation requires only $q$ inversions of size $n \times n$, which is a significant improvement over $\mathcal{O}(n^3p^3)$ computations.

\section{Numerical illustrations} \label{s:illustration}

The two synthetic examples separate trend identification from prediction. The first uses a data-generating function with a known low-order trend to evaluate the effect of orthogonality under uniform and clustered input designs. The second uses the multi-fidelity test functions of \citet{forrester2008} in a $2\times2$ ablation of orthogonality and output sharing.

\subsection{Recovery of a known low-order trend} \label{s:illustration.trend}

Most benchmark functions are not specified as the sum of a low-order trend and an orthogonal deviation, so they do not provide target regression coefficients. The following construction provides such a target.

Let $\mathcal{X}=[-1,1]$, let $\vect{g}(x)=(1,x)\transpose$, and let $\mathcal{G}=\mathrm{span}\{1,x\}$ be the corresponding regression space. Begin with
$s(x)=\sin(-\pi x)$. Its $L^2(\mathcal X)$ projection onto $\mathcal G$ is
$\Pi_{\mathcal G}s(x)=-3x/\pi$.
Subtracting this projection and rescaling the residual to unit norm gives
\begin{align} \label{eq:wiggle}
w(x) := \kappa\inv \left(\sin(-\pi x) + \frac{3x}{\pi}\right), \qquad
\kappa^2 = 1-\frac{6}{\pi^2},
\end{align}
which satisfies $\int_{-1}^1 w(x)\,dx=\int_{-1}^1 x\,w(x)\,dx=0$, so $w\in\mathcal{G}^\perp$. The $p=3$ outputs are generated as
\begin{align} \label{eq:trend_generator}
y_\ell(x) = \underbrace{a_\ell + b_\ell x}_{\text{trend}} + \underbrace{\alpha_\ell\, w(x)}_{\text{deviation}} + \eps_\ell,
\qquad \eps_\ell \sim \mathcal{N}(0,\sigma^2_\ell), \qquad \ell=1,2,3,
\end{align}
with $\vect{a}=(2,8,-3)$, $\vect{b}=(12,-10,6)$, $\vect{\alpha}=(2.5,2.0,1.5)$, and $\vect{\sigma}=(0.5,0.4,0.3)$.

Because $w\in\mathcal{G}^\perp$, the pair $(a_\ell,b_\ell)$ gives the coefficients of the $L^2(\mathcal X)$ projection of the noise-free signal onto $\mathcal{G}$. Thus, any difference between a fitted coefficient and $(a_\ell,b_\ell)$ is attributable to estimation rather than ambiguity in the target trend. Inputs are generated directly on $[-1,1]$, and further input standardization is disabled, so the generator and MOOGP use the same fixed domain, basis, and Lebesgue measure.

Two input designs are used. The \emph{uniform} design is a one-dimensional Latin hypercube sample on $[-1,1]$, which distributes points evenly over the domain. The \emph{clustered} design applies the same stratified sampling on a logarithmic scale and then maps the points to $[-1,1]$. This design concentrates points near $x=1$ and leaves values near $x=-1$ sparsely covered. Such uneven coverage is common when runs accumulate in regions of interest rather than filling the domain.

Each design contains $n=60$ points. Three trend estimators are fit: MOOGP, the non-orthogonal MOGP obtained by replacing each $c_k^*$ with $c_k$, and ordinary least squares of $\vect{y}$ on $\vect{g}(x)$. Both GP models use $\vect{g}(x)=(1,x)\transpose$ and $q=p=3$. Orthogonalization is the only difference between them.

\begin{figure}[!t]
    \centering
    \includegraphics[width=1\linewidth]{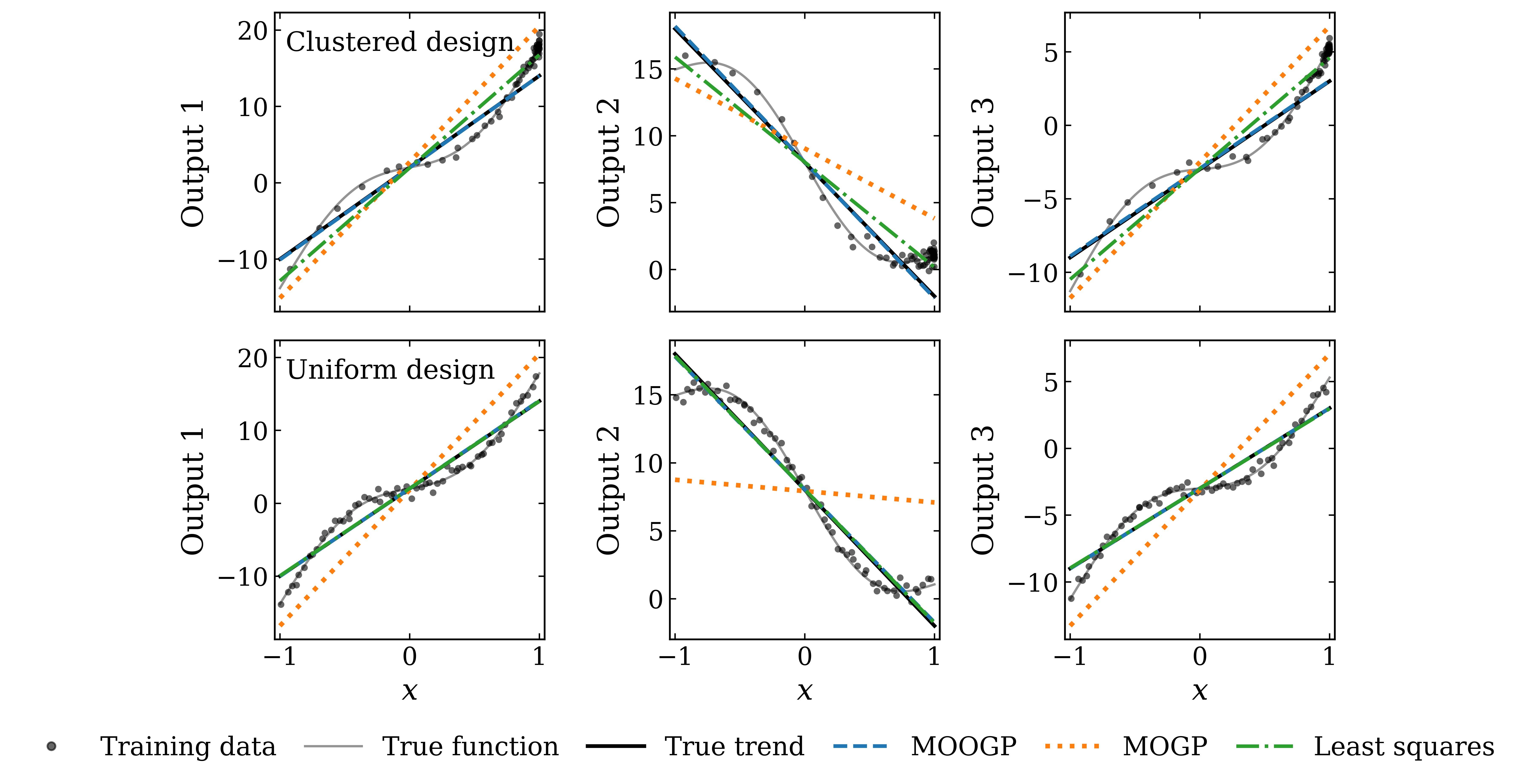}
    \caption{Trend recovery for the data-generating function in \Cref{eq:trend_generator} under the clustered design (top row) and uniform design (bottom row). Each output is a line plus a smooth deviation orthogonal to $\mathrm{span}\{1,x\}$, so the solid black line is both the data-generating trend and the $L^2(\mathcal X)$ projection of the noise-free signal. One replication is shown.}
    \label{fig:low_order_trend}
\end{figure}

\begin{table}[t]
\centering
\caption{Estimated slopes $\hat{b}_\ell$ for the function in \Cref{eq:trend_generator}. Entries are means over 100 replications, with standard deviations across replications in parentheses. The final column reports the mean absolute error in the slope coefficients, $(100p)^{-1}\sum_{r=1}^{100}\sum_{\ell=1}^p|\hat b_{r\ell}-b_\ell|$, over the replications and $p=3$ outputs.}
\label{tab:trend_recovery}
\small
\setlength{\tabcolsep}{5pt}
\renewcommand{\arraystretch}{1.15}
\begin{tabular}{llccc c}
\toprule
\textbf{Design} & \textbf{Model} & \textbf{Output 1} & \textbf{Output 2} & \textbf{Output 3} & \textbf{Mean abs.\ error} \\
\midrule
\multirow{3}{*}{\textbf{Clustered}}
& MOOGP         & $12.04 \, (0.19)$ & $-10.00 \, (0.15)$ & $ 5.98 \, (0.12)$ & $ 0.12$ \\
& MOGP          & $17.37 \, (1.14)$ & $-4.61 \, (1.57)$ & $ 9.28 \, (0.72)$ & $ 4.68$ \\
& Least squares & $14.72 \, (0.19)$ & $-7.85 \, (0.13)$ & $ 7.61 \, (0.11)$ & $ 2.16$ \\
\midrule
\multirow{3}{*}{\textbf{Uniform}}
& MOOGP         & $11.99 \, (0.12)$ & $-9.99 \, (0.09)$ & $ 6.01 \, (0.07)$ & $ 0.07$ \\
& MOGP          & $18.01 \, (1.11)$ & $-4.10 \, (1.47)$ & $ 9.81 \, (0.70)$ & $ 5.24$ \\
& Least squares & $12.00 \, (0.12)$ & $-9.99 \, (0.09)$ & $ 6.01 \, (0.07)$ & $ 0.07$ \\
\bottomrule
\end{tabular}
\end{table}

\Cref{fig:low_order_trend} shows one replication under each design, and \Cref{tab:trend_recovery} summarizes 100 replications over independent designs and noise draws. The non-orthogonal MOGP trend differs substantially from the target under both designs. Its mean estimated slope for output 1 is $17.37$, compared with the true value of $12$, and its mean estimated slope for output 2 is $-4.61$, compared with the true value of $-10$. The mean absolute slope error is $4.68$ under the clustered design and $5.24$ under the uniform design. Thus, in this example, the confounding cannot be attributed to uneven domain coverage. It arises from allowing the latent covariance to contain components in the regression space.

Least squares behaves differently. Under the uniform design, its mean absolute slope error is $0.07$. Under the clustered design, the error increases to $2.16$, although the slope signs remain correct. A linear trend estimator is biased by the component of $w$ that its weighting of the data projects onto $\mathcal{G}$. For least squares, this weighting corresponds to the empirical distribution of the design points. The relevant projection is therefore evaluated at the observed points rather than with respect to the Lebesgue measure on $\mathcal{X}$. A space-filling design approximates this measure, with respect to which $w$ is orthogonal, so the bias nearly vanishes. Under the clustered design, the mean slope of the discrete projection across replications is $1.08$. Multiplication by $\vect{\alpha}$ predicts mean signed slope biases of $(2.70, 2.16, 1.62)$ across the three outputs, which agree with the observed biases to within $0.02$.

MOOGP remains close to the target under both designs, with mean absolute slope errors of $0.12$ and $0.07$. Under the uniform design, its error is comparable to the least-squares error of $0.07$. Orthogonality prevents the latent covariance from representing components in the trend space, but it does not force a finite-sample generalized least-squares estimate from noisy data to equal the population $L^2(\mathcal X)$ projection. The experiment does not decompose the small remaining error by source. In this example, MOOGP estimates the trend jointly with the emulator and remains close to the target without requiring the input design to approximate the orthogonality measure.

\subsection{Separating orthogonality and output sharing} \label{s:illustration.forrester}

Orthogonality separates the regression component from the residual process,
whereas output sharing permits borrowing across correlated responses. A
matched $2\times2$ ablation separates these effects using the multi-fidelity
test functions of \citet{forrester2008}. Consider
$\vect{f}(x)=(f_1(x),f_2(x),f_3(x))^\top$, $x\in[0,1]$, where
\begin{align*}
f_1(x) &= (6x-2)^2\sin(12x-4), \\
f_2(x) &= 0.5f_1(x)+5(x-0.5)+5, \\
f_3(x) &= -0.8f_1(x)-5(x-0.5)-4.
\end{align*}
The outputs share a nonlinear component but have different linear trends.
Independent Gaussian errors are added with covariance
$\Sigerror=\diag{10,1,0.05}$.

The four models are three independently fit scalar GPs, three independently
fit scalar OGPs, MOGP, and MOOGP. Each scalar fit uses $p=q=1$, and the joint
models use $p=q=3$. All models use $\vect{g}(x)=(1,x)\transpose$, the same
kernel family, and the same fitting procedure. The GP--OGP and MOGP--MOOGP
comparisons therefore isolate orthogonality, while the independent--joint
comparisons isolate output sharing. Uniform and clustered training designs
contain $n=25$ inputs, with the clustered design concentrating inputs near
$x=1$. The experiment uses 200 paired replications. Within each replication,
all models and both designs are evaluated on the same independent uniform test
set of size $n'=500$. All fits converged.

The noise-free functions have different scales, so point prediction is
measured by standardized recovery RMSE. Let $\tau_\ell^2=
\int_0^1\{f_\ell(x)-\bar f_\ell\}^2\,dx$, where
$\bar f_\ell=\int_0^1f_\ell(x)\,dx$. Then
\begin{align*}
\mathrm{sRRMSE}
&=\left[
\frac{1}{n'p}\sum_{i=1}^{n'}\sum_{\ell=1}^p
\left\{
\frac{\mu_{\mathrm{pred},\ell}(x_i')-f_\ell(x_i')}{\tau_\ell}
\right\}^2
\right]^{1/2}.
\end{align*}
Trend RMSE is the integrated error between the fitted trend and the population
$L^2[0,1]$ linear projection of each output. Marginal coverage and the
multivariate Dawid--Sebastiani score, DSS$_{\mathrm{joint}}$, evaluate noisy
test responses \citep{gneiting2007strictly}. The latter uses the full
pointwise output covariance; this covariance is diagonal for the independent
models.

\begin{table}[t]
\centering
\caption{Forrester $2\times2$ ablation under uniform and clustered training
designs. Entries are means over 200 paired replications. Lower values are
preferred for sRRMSE, trend RMSE, and DSS$_{\mathrm{joint}}$; coverage is
assessed relative to $0.95$.}
\label{tab:model_comparison}
\small
\setlength{\tabcolsep}{5pt}
\renewcommand{\arraystretch}{1.15}
\begin{tabular}{
  ll
  S[table-format=1.3]
  S[table-format=1.3]
  S[table-format=1.3]
  S[table-format=1.3]
}
\toprule
\textbf{Design} & \textbf{Model} & \textbf{sRRMSE} &
\textbf{Trend RMSE} & \textbf{Coverage} &
\textbf{DSS$_{\mathrm{joint}}$} \\
\midrule
\multirow{4}{*}{\textbf{Uniform}}
& Independent GP  & 0.274 & 1.939 & 0.937 & 4.643 \\
& Independent OGP & 0.259 & 0.483 & 0.928 & 4.677 \\
& MOGP            & 0.211 & 1.721 & 0.923 & 5.562 \\
& MOOGP           & 0.217 & 0.479 & 0.920 & 5.605 \\
\midrule
\multirow{4}{*}{\textbf{Clustered}}
& Independent GP  & 0.522 & 2.611 & 0.947 & 7.571 \\
& Independent OGP & 0.514 & 1.455 & 0.926 & 7.778 \\
& MOGP            & 0.439 & 2.351 & 0.935 & 7.339 \\
& MOOGP           & 0.462 & 1.457 & 0.927 & 7.708 \\
\bottomrule
\end{tabular}
\end{table}

Output sharing reduces sRRMSE under both designs and both orthogonality
conditions. Relative to the corresponding independent model, the reduction is
23\% for MOGP and 16\% for MOOGP under the uniform design, and 16\% and 10\%
under the clustered design. All four paired 95\% confidence intervals exclude
zero. The gain is concentrated in the noisiest output, for which sharing
reduces sRRMSE by 20--27\%. Sharing increases error for the lowest-noise output
under the uniform design, indicating some negative transfer.

Orthogonality substantially reduces trend RMSE in both the independent and
joint comparisons. In contrast, MOOGP has slightly higher sRRMSE than MOGP
under both designs; the paired mean differences are $0.006$ with a 95\%
confidence interval of $(0.004,0.009)$ under the uniform design and $0.023$
with an interval of $(0.008,0.038)$ under the clustered design. Clustering
increases prediction error for every model, and MOGP has the lowest mean
sRRMSE and joint DSS in that setting. MOOGP nevertheless retains a predictive
advantage over independent OGP and recovers the population trend more
accurately than MOGP. In this example, the ablation therefore attributes the
main predictive gain to output sharing and the main trend-identification gain
to orthogonality, with a modest predictive cost for combining the two.

\section{Numerical experiment and application} \label{s:numerical}

\subsection{Multi-output Borehole benchmark}\label{s:borehole}

The experiment uses a multi-output modification of the Borehole benchmark, which
models water flow between two aquifers \citep{morris1993borehole}. Four simulator
inputs are varied, and the adapted response is evaluated at $p=5$ or $10$
two-dimensional outputs. Independent Gaussian errors are added to both
training and test responses. For each output, the error variance is $5\%$ of the
sample variance of its noise-free training response. Consequently, the reported
RMSE and coverage evaluate prediction of future noisy observations.

Training sizes are $n\in\{50,100,250,1000,2500\}$. Each $(n,p)$ setting
has five training-design replications paired across methods, and every
replication has an independently generated 800-run test design. For MOOGP, its
non-orthogonal MOGP ablation, LCGP \citep{chan2023highdimensional}, and OILMM
\citep{bruinsma2020scalable}, a common latent rank is selected from the
columnwise standardized noisy training response as the smallest rank whose
cumulative squared singular values exceed $99\%$ of the total. PUQ fits
independent heteroskedastic GPs through its \texttt{multihetGP} implementation
\citep{surer2025puq,binois2018practical,ogara2025hetgpy}. All methods use a 1000-iteration
limit. The runs were performed on a 32-vCPU AWS node with an Intel Xeon Platinum
8488C processor and 64 GiB of memory. Up to four fits ran concurrently, each
pinned to four dedicated physical cores with a four-thread numerical-backend
budget.

\Cref{fig:borehole-performance} compares test RMSE, empirical coverage of the marginal $95\%$ prediction intervals, and training time.
\begin{figure}[!htbp]
    \centering
    \includegraphics[width=\linewidth]{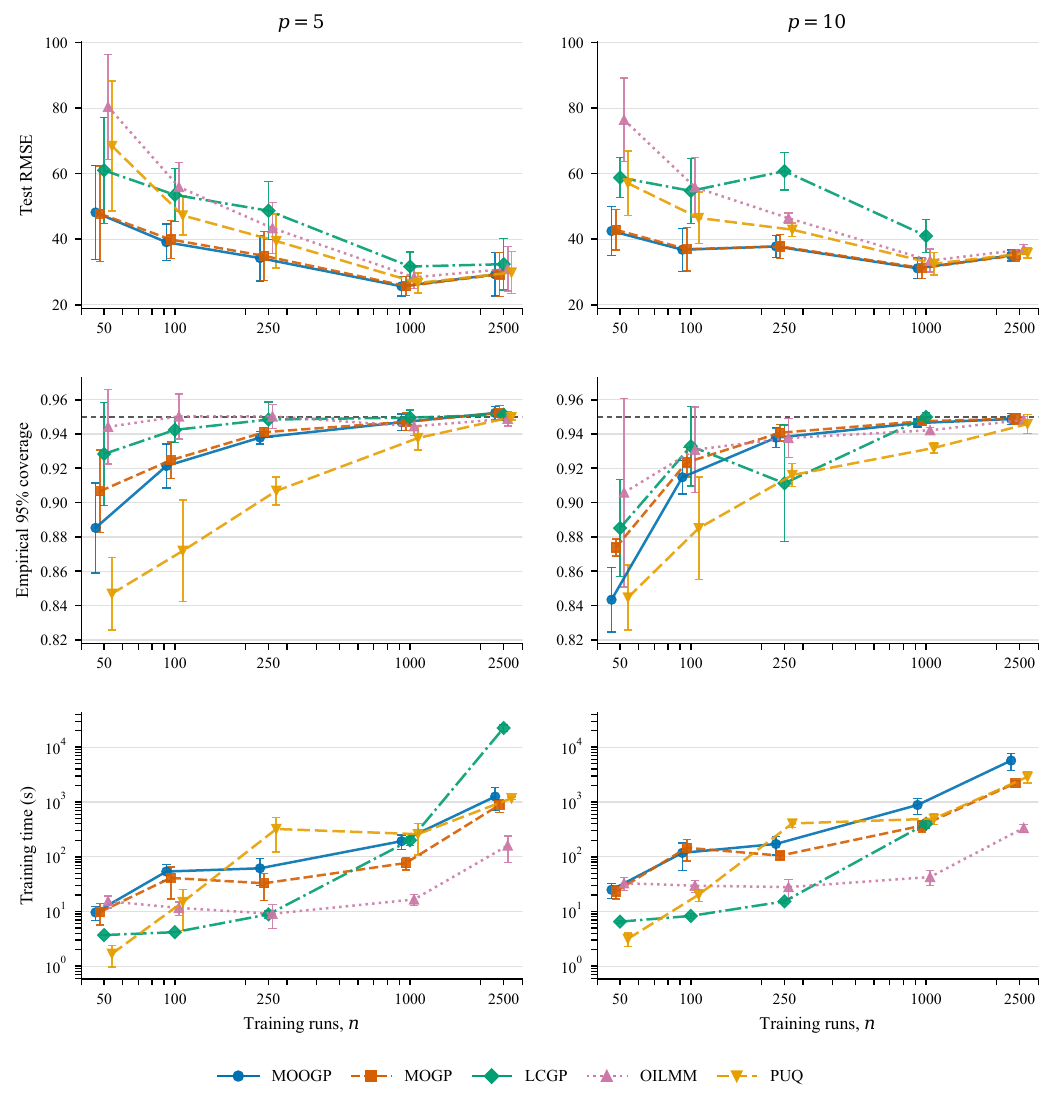}
    \caption{Borehole test performance across five paired replications for each $(n,p)$ setting. Faint marks show individual replications; large marks and bars show the mean and $\pm 1$ sample standard deviation. The dashed line marks the nominal $0.95$ coverage. Training size and training time use logarithmic scales. LCGP has no result for $n=2500,p=10$.}
    \label{fig:borehole-performance}
\end{figure}
The paired MOOGP--MOGP comparison does not show a consistent predictive advantage from orthogonality. One of these two models has the lowest mean RMSE in every setting, and their errors become nearly indistinguishable at the largest training sizes. Both methods approach nominal coverage as $n$ increases. OILMM is closer to nominal coverage in several small-sample settings but has substantially larger RMSE, while PUQ under-covers at small and intermediate training sizes.

PUQ is fastest at $n=50$, while LCGP is the fastest latent-factor method through
$n=250$. LCGP was not run for $n=2500,p=10$ after its $n=2500,p=5$ fits
required more than 22,000 seconds on average. OILMM is fastest at $n=1000$ and
$2500$. MOOGP is slower than MOGP for large $n$ and $p$, consistent with the
additional orthogonality calculations. PUQ returns optimizer convergence
warnings in most fits, which qualify its reported results.

\subsection{Viscous anisotropic hydrodynamics application}\label{s:application}

Viscous anisotropic hydrodynamics (VAH) is a high-fidelity model used in Bayesian calibration of relativistic heavy-ion collisions \citep{mcnelis20183+,mcnelis2021anisotropic,liyanage2023bayesian}. The data analyzed here consist of 541 VAH runs for Pb-Pb collisions at $\sqrt{s_{\mathrm{NN}}}=2.76$ TeV, with 15 input parameters and 98 observables organized into 11 multiplicity, transverse-energy, mean-transverse-momentum, and flow-harmonic families \citep{liyanage2023bayesian}. These dimensions make direct probabilistic calibration expensive and motivate a surrogate that can predict the full response vector while retaining interpretable low-order input effects.

All methods use the same deterministic five-fold partition. Each method's native
optimizer is configured with a maximum-iteration setting of \num{1000}, although
iteration counts are not directly comparable across implementations. The latent
rank is $q=5$ where applicable. The MOOGP and MOGP trends contain an intercept
and the 15 standardized main effects. MOOGP, MOGP, and LCGP use an 11-group
diagonal-noise specification specified by the observable families; OILMM and
PUQ do not expose an equivalent parameterization. The training for the VAH application is performed on an Apple-silicon MacBook Pro.

For evaluation, prediction errors and marginal predictive standard deviations
are divided by the corresponding sample standard deviation computed from the
training outputs in each fold. Standardized RMSE, interval width,
DSS per output, and coverage are then pooled over all
$541\times98$ held-out scalar predictions, so each simulator run contributes
once. Lower values are preferred except that coverage is assessed relative to
$0.95$. Raw RMSE is omitted because the 98 observables have different physical
scales.

\Cref{tab:vah-cross-validation} reports the pooled predictive metrics and fold-mean training times. Note that the timing results between the Borehole example and VAH application are not comparable due to different computational environments.
\begin{table}[!htbp]
\centering
\caption{Five-fold VAH cross-validation. Prediction metrics pool the held-out predictions, with output scales estimated from the corresponding training fold. Standardized RMSE, interval width, and DSS per output are better when smaller; coverage is assessed relative to $0.95$. Training time is the fold mean in seconds with sample standard deviation in parentheses.}
\label{tab:vah-cross-validation}
\small
\setlength{\tabcolsep}{4pt}
\renewcommand{\arraystretch}{1.12}
\begin{tabular}{@{}lrrrrr@{}}
\toprule
Method & sRMSE & Coverage & sWidth & sDSS & Time (SD) \\
\midrule
MOOGP & 0.413 & 0.922 & 1.431 & -0.758 & 87.7 (7.4) \\
MOGP  & 0.439 & 0.928 & 1.568 & -0.655 & 67.4 (9.5) \\
LCGP  & 0.846 & 0.900 & 3.075 &  0.961 & 24.3 (4.8) \\
OILMM & 0.776 & 0.923 & 2.843 &  0.536 & 10.8 (2.8) \\
PUQ   & 0.407 & 0.898 & 1.261 & -0.963 & 1175.8 (87.6) \\
\bottomrule
\end{tabular}
\end{table}
PUQ has the lowest standardized RMSE, narrowest intervals, and lowest DSS, but it under-covers and all five folds return optimizer convergence warnings. In this application, MOOGP has $6\%$ lower standardized RMSE, $9\%$ narrower intervals, and a lower DSS than MOGP, with slightly lower coverage and a longer training time. LCGP and OILMM are faster but have substantially worse standardized RMSE, width, and DSS.

\Cref{fig:vah-group-performance} reports standardized RMSE and coverage for the 11 observable families.
\begin{figure}[!htbp]
    \centering
    \includegraphics[width=\linewidth]{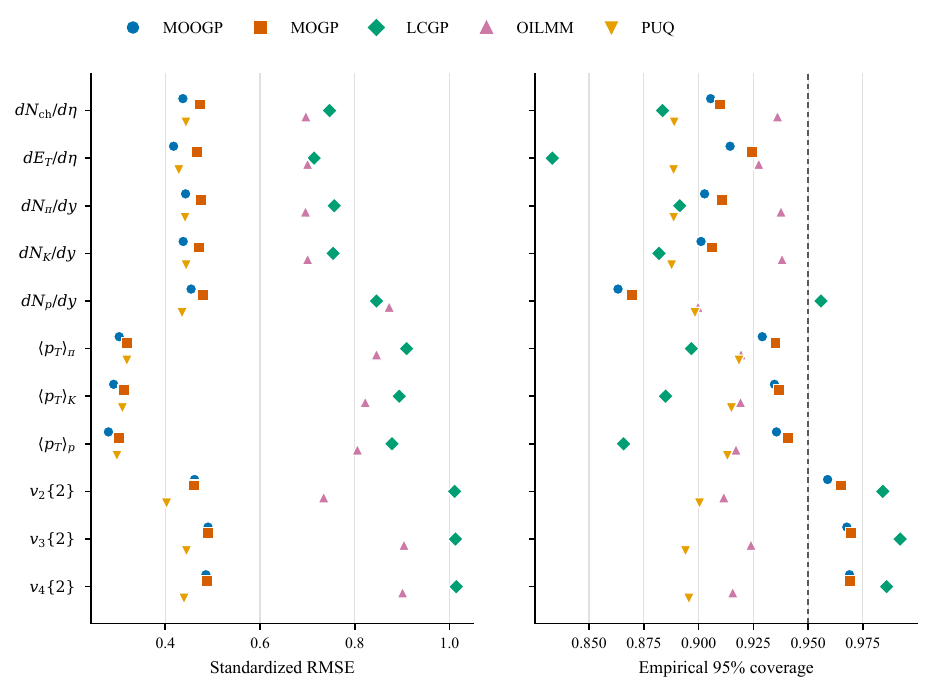}
    \caption{VAH five-fold cross-validation by observable family. Each mark pools all held-out runs and every output in the indicated family after standardization by the corresponding training-fold output scale. Standardized RMSE is better when smaller. The dashed vertical line in the coverage panel marks the nominal value $0.95$. The groups contain 98 observables in total and are defined before model fitting.}
    \label{fig:vah-group-performance}
\end{figure}
Within this cross-validation, MOOGP has lower standardized RMSE than MOGP in ten of the eleven groups and is most accurate for the three mean-$p_T$ families. Proton yield is the principal source of its undercoverage. PUQ is more accurate for the flow harmonics but under-covers in every group, while LCGP shows the greatest variation in coverage and the widest flow-harmonic intervals.

\Cref{fig:vah-trend-coefficients} compares the fold-mean standardized MOOGP coefficients with those from MOGP.
\begin{figure}[!htbp]
    \centering
    \includegraphics[width=\linewidth]{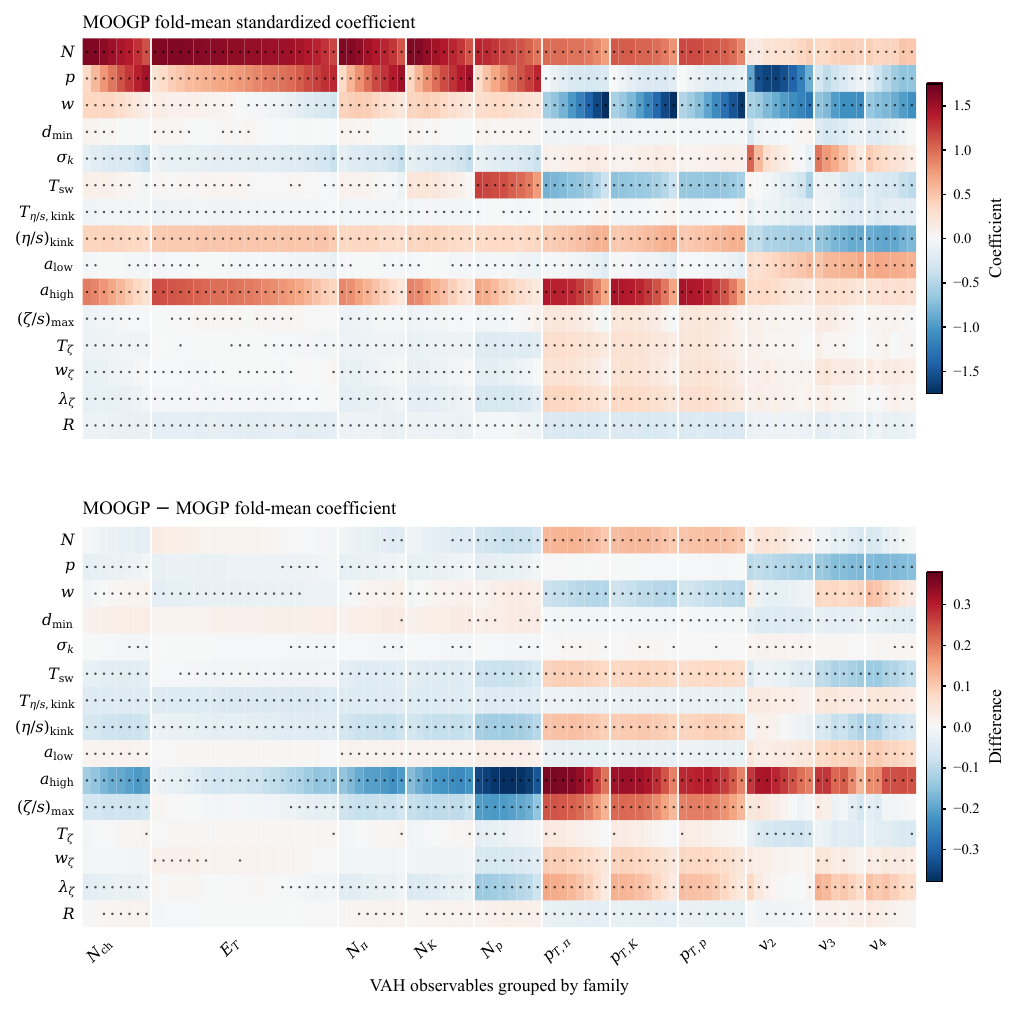}
    \caption{Standardized low-order trend coefficients for the VAH application. Rows are the 15 simulator inputs and columns are the 98 outputs, with vertical boundaries separating the 11 observable families. The upper panel shows the MOOGP fold mean; the lower panel shows the fold mean of MOOGP minus MOGP. Each panel has a separate symmetric color scale fixed at its largest absolute value. A black dot marks a coefficient whose sign is the same in at least four of the five folds. The intercept is omitted.}
    \label{fig:vah-trend-coefficients}
\end{figure}
Several estimated associations via the regression are stable across folds. The normalization parameter $N$ is positively associated with multiplicity and transverse energy, while the nucleon width $w$ is negatively associated with mean $p_T$ and the generalized-mean parameter $p$ with central $v_2\{2\}$. The largest stable MOOGP--MOGP differences involve the high-temperature shear-viscosity slope $a_{\mathrm{high}}$. While these fitted low-order associations should not be interpreted as causal physical effects and are specific to this set of VAH simulation data, the specific difference, most notably in $a_{\mathrm{high}}$ and $(\zeta/s)_\text{max}$, resulted from enforcing orthogonoality demonstrates interpretative quality of the proposed method.

\section{Conclusion} \label{s:conclusions_discussion}
We introduced the multi-output orthogonal Gaussian process model in this paper as an effective surrogate model for stochastic simulators.  The model consists of a deterministic flexible regression mean and a multi-output GP residual, where the residual is constructed via independent latent GP components.  The resulting model extends orthogonal GP to the vector-valued setting by restricting each latent component to lie in the orthogonal complement of the regression space.  

The numerical studies show when orthogonality matters.  Under a well-spread space-filling design, MOOGP and its non-orthogonal counterpart perform comparably. The constraint costs little when it is not needed. Under uneven input coverage, the non-orthogonal model lets the GP residual absorb variation belonging to the regression space, whereas MOOGP recovers the true trend and produces sharper predictions.  The heavy-ion collision case study confirms that these gains persist at realistic scale, with a 98-dimensional output emulated through a low-dimensional latent basis, and the timing comparisons show the proposed computations remain competitive as the number of outputs grows.

Follow-up studies are facilitated by this methodological development, including Bayesian parameter calibration of the VAH model.  MOOGP can study structural interaction and higher-order regression terms in the physics model \citep{ehlers2026varp}.  Since the orthogonalization decoupling identifies reliable trend coefficients and provides sharper predictions, it can strengthen identifiability between parameters and any discrepancy terms within multi-output model calibration \cite{kennedy2001bayesian,plumlee2017bayesian,tuo2019adjustments,xie2021bayesian}.
In our development, stationary covariance functions are primarily studied.  Similar orthogonality conditions may be developed for other covariances that addresses the mean-regressing behavior or boundary information in GPs \citep{tan2018gaussian,ding2019bdrygp,joseph2025rational}.  Extensions to input-dependent noise, non-Gaussian outputs, and fully Bayesian hyperparameter inference would widen the class of stochastic simulators the model can serve.

\section*{Disclosure statement}
The authors report there are no competing interests to declare.

\if0\blind{
\section*{Acknowledgments}
The authors acknowledge the support from the National Science Foundation under Grant No.~OAC-2004601. ECB acknowledges the additional support from the Northwestern University Office of Undergraduate Research.  The authors thank Dan Liyanage and Uli Heinz for approving the use of the heavy-ion collision simulation data analyzed in \Cref{s:application} and are grateful for the insightful physics discussions with Dan Liyanage and Sunil Jaiswal.  This research was supported in part through the computational resources and staff contributions provided for the Quest high performance computing facility at Northwestern University, which is jointly supported by the Office of the Provost, the Office for Research, and Northwestern University Information Technology.} \fi

\section*{Declaration of generative AI use} During the preparation of this manuscript, the authors used Claude (Anthropic, Claude Opus/Fable 5) and ChatGPT (OpenAI, GPT-5-sol) to check consistency of mathematical notation and writing style across sections, and to assist in writing plotting scripts for the figures. All outputs were reviewed, verified, and edited by the authors, who take full responsibility for the content of the publication.

\bibliographystyle{chicago}
\spacingset{1}
\bibliography{moogp}

@book{harville2008matrix,
  title={Matrix Algebra From a Statistician's Perspective},
  author={Harville, David A.},
  year={2008},
  publisher={Springer},
  address={New York},
  isbn={978-0-387-78356-7}
}

@article{plumlee2018orthogonal,
  title={Orthogonal Gaussian process models},
  author={Plumlee, Matthew and Joseph, V Roshan},
  journal={Statistica Sinica},
  volume={28},
  number={2},
  pages={601--619},
  year={2018},
  doi={10.5705/ss.202015.0404}
}

@article{morris1993borehole,
  title={Bayesian Design and Analysis of Computer Experiments: Use of Derivatives in Surface Prediction},
  author={Morris, Max D. and Mitchell, Toby J. and Ylvisaker, Donald},
  journal={Technometrics},
  volume={35},
  number={3},
  pages={243--255},
  year={1993},
  doi={10.1080/00401706.1993.10485320}
}

@inbook{forrester2008,
author = {Forrester, Sóbester, Keane},
publisher = {John Wiley \& Sons, Ltd},
isbn = {9780470770801},
title = {Multi-Fidelity Analysis},
booktitle = {Engineering Design via Surrogate Modelling},
chapter = {8},
pages = {167-177},
doi = {https://doi.org/10.1002/9780470770801.ch8},
url = {https://onlinelibrary.wiley.com/doi/abs/10.1002/9780470770801.ch8},
eprint = {https://onlinelibrary.wiley.com/doi/pdf/10.1002/9780470770801.ch8},
year = {2008}
}

@misc{ding2019bdrygp,
  title={BdryGP: a new Gaussian process model for incorporating boundary information},
  author={Ding, Liang and Mak, Simon and Wu, CF},
  journal={arXiv preprint arXiv:1908.08868},
  year={2019}
}

@article{tan2018gaussian,
  title={Gaussian process modeling with boundary information},
  author={Tan, Matthias Hwai Yong},
 journal = {Statistica Sinica},
 number = {2},
 pages = {621--648},
 publisher = {Institute of Statistical Science, Academia Sinica},
 volume = {28},
 year = {2018}
}

@article{joseph2025rational,
  title={Rational kriging},
  author={Joseph, V Roshan},
  journal={Journal of the American Statistical Association},
  volume={120},
  number={549},
  pages={548--558},
  year={2025},
  publisher={Taylor \& Francis}
}

@book{hardle2015multivariate,
  title={Multivariate statistics: exercises and solutions},
  author={H{\"a}rdle, Wolfgang and Hl{\'a}vka, Zden{\v{e}}k},
  year={2015},
  publisher={Springer}
}

@article{sacks1989design,
  title={Design and Analysis of Computer Experiments},
  author={Sacks, Jerome and Welch, William J. and Mitchell, Toby J. and Wynn, Henry P.},
  journal={Statistical Science},
  volume={4},
  number={4},
  pages={409--435},
  year={1989},
  doi={10.1214/ss/1177012413}
}

@book{santner2003design,
  title={The Design and Analysis of Computer Experiments},
  author={Santner, Thomas J. and Williams, Brian J. and Notz, William I.},
  series={Springer Series in Statistics},
  year={2003},
  publisher={Springer},
  address={New York},
  isbn={978-0-387-95420-2},
  doi={10.1007/978-1-4757-3799-8}
}

@book{joseph2026experimental,
  title={Experimental Design for Data Science and Engineering},
  author={Joseph, V. Roshan},
  series={Chapman \& Hall/CRC Texts in Statistical Science},
  year={2026},
  publisher={Chapman \& Hall/CRC},
  isbn={978-1-0411-1752-0},
  doi={10.1201/9781003661535}
}

@article{kennedy2001bayesian,
  title={Bayesian Calibration of Computer Models},
  author={Kennedy, Marc C. and O'Hagan, Anthony},
  journal={Journal of the Royal Statistical Society: Series B (Statistical Methodology)},
  volume={63},
  number={3},
  pages={425--464},
  year={2001},
  doi={10.1111/1467-9868.00294}
}

@book{rasmussen2006gaussian,
  title={Gaussian Processes for Machine Learning},
  author={Rasmussen, Carl Edward and Williams, Christopher K. I.},
  series={Adaptive Computation and Machine Learning},
  year={2006},
  publisher={MIT Press},
  address={Cambridge, MA},
  isbn={978-0-262-18253-9},
  url={https://gaussianprocess.org/gpml/}
}

@book{gramacy2020surrogates,
  title={Surrogates: Gaussian Process Modeling, Design, and Optimization for the Applied Sciences},
  author={Gramacy, Robert B.},
  series={Chapman \& Hall/CRC Texts in Statistical Science},
  year={2020},
  publisher={CRC Press},
  address={Boca Raton, FL},
  isbn={978-0-367-41542-6},
  doi={10.1201/9780367815493}
}

@article{ankenman2010stochastic,
  title={Stochastic Kriging for Simulation Metamodeling},
  author={Ankenman, Bruce and Nelson, Barry L. and Staum, Jeremy},
  journal={Operations Research},
  volume={58},
  number={2},
  pages={371--382},
  year={2010},
  doi={10.1287/opre.1090.0754}
}

@article{gramacy2012cases,
  title={Cases for the Nugget in Modeling Computer Experiments},
  author={Gramacy, Robert B. and Lee, Herbert K. H.},
  journal={Statistics and Computing},
  volume={22},
  number={3},
  pages={713--722},
  year={2012},
  doi={10.1007/s11222-010-9224-x}
}

@article{baker2022analyzing,
  title={Analyzing Stochastic Computer Models: A Review with Opportunities},
  author={Baker, Evan and Barbillon, Pierre and Fadikar, Arindam and Gramacy, Robert B. and Herbei, Radu and Higdon, David and Huang, Jiangeng and Johnson, Leah R. and Ma, Pulong and Mondal, Anirban and Pires, Bianica and Sacks, Jerome and Sokolov, Vadim},
  journal={Statistical Science},
  volume={37},
  number={1},
  pages={64--89},
  year={2022},
  doi={10.1214/21-STS822}
}

@article{fadikar2018calibrating,
  title={Calibrating a Stochastic, Agent-Based Model Using Quantile-Based Emulation},
  author={Fadikar, Arindam and Higdon, Dave and Chen, Jiangzhuo and Lewis, Bryan and Venkatramanan, Srinivasan and Marathe, Madhav},
  journal={SIAM/ASA Journal on Uncertainty Quantification},
  volume={6},
  number={4},
  pages={1685--1706},
  year={2018},
  doi={10.1137/17M1161233}
}

@article{marrel2012global,
  title={Global Sensitivity Analysis of Stochastic Computer Models with Joint Metamodels},
  author={Marrel, Amandine and Iooss, Bertrand and Da Veiga, S{\'e}bastien and Ribatet, Mathieu},
  journal={Statistics and Computing},
  volume={22},
  number={3},
  pages={833--847},
  year={2012},
  doi={10.1007/s11222-011-9274-8}
}

@article{herbei2014estimating,
  title={Estimating Ocean Circulation: An {MCMC} Approach with Approximated Likelihoods via the {Bernoulli} Factory},
  author={Herbei, Radu and Berliner, L. Mark},
  journal={Journal of the American Statistical Association},
  volume={109},
  number={507},
  pages={944--954},
  year={2014},
  doi={10.1080/01621459.2014.914439}
}

@article{binois2018practical,
  title={Practical Heteroscedastic Gaussian Process Modeling for Large Simulation Experiments},
  author={Binois, Micka{\"e}l and Gramacy, Robert B. and Ludkovski, Mike},
  journal={Journal of Computational and Graphical Statistics},
  volume={27},
  number={4},
  pages={808--821},
  year={2018},
  doi={10.1080/10618600.2018.1458625}
}

@article{conti2010bayesian,
  title={Bayesian Emulation of Complex Multi-Output and Dynamic Computer Models},
  author={Conti, Stefano and O'Hagan, Anthony},
  journal={Journal of Statistical Planning and Inference},
  volume={140},
  number={3},
  pages={640--651},
  year={2010},
  doi={10.1016/j.jspi.2009.08.006}
}

@article{higdon2008computer,
  title={Computer Model Calibration Using High-Dimensional Output},
  author={Higdon, Dave and Gattiker, James and Williams, Brian and Rightley, Maria},
  journal={Journal of the American Statistical Association},
  volume={103},
  number={482},
  pages={570--583},
  year={2008},
  doi={10.1198/016214507000000888}
}

@article{phillips2021band,
  title={Get on the {BAND} Wagon: A {Bayesian} Framework for Quantifying Model Uncertainties in Nuclear Dynamics},
  author={Phillips, Daniel R. and Furnstahl, Richard J. and Heinz, Ulrich and Maiti, Tapabrata and Nazarewicz, Witold and Nunes, Filomena M. and Plumlee, Matthew and Pratola, Matthew T. and Pratt, Scott and Viens, Frederi G. and Wild, Stefan M.},
  journal={Journal of Physics G: Nuclear and Particle Physics},
  volume={48},
  number={7},
  pages={072001},
  year={2021},
  doi={10.1088/1361-6471/abf1df}
}

@article{liyanage2023bayesian,
  title={Bayesian Calibration of Viscous Anisotropic Hydrodynamic Simulations of Heavy-Ion Collisions},
  author={Liyanage, Dananjaya and S{\"u}rer, {\"O}zge and Plumlee, Matthew and Wild, Stefan M. and Heinz, Ulrich},
  journal={Physical Review C},
  volume={108},
  number={5},
  pages={054905},
  year={2023},
  doi={10.1103/PhysRevC.108.054905}
}

@article{gu2016parallel,
  title={Parallel Partial Gaussian Process Emulation for Computer Models with Massive Output},
  author={Gu, Mengyang and Berger, James O.},
  journal={The Annals of Applied Statistics},
  volume={10},
  number={3},
  pages={1317--1347},
  year={2016},
  doi={10.1214/16-AOAS934}
}

@article{ogara2025hetgpy,
  title={{hetGPy}: Heteroskedastic Gaussian Process Modeling in Python},
  author={O'Gara, David and Binois, Micka{\"e}l and Garnett, Roman and Hammond, Ross A.},
  journal={Journal of Open Source Software},
  volume={10},
  number={106},
  pages={7518},
  year={2025},
  doi={10.21105/joss.07518}
}

@phdthesis{chan2023highdimensional,
  title={High-Dimensional Gaussian Process Methods for Uncertainty Quantification},
  author={Chan, Moses Y.-H.},
  school={Northwestern University},
  address={Evanston, IL},
  year={2023},
  month={12},
  note={ProQuest Dissertations and Theses, 30693493}
}

@article{chan2024constructing,
  title={Constructing a Simulation Surrogate with Partially Observed Output},
  author={Chan, Moses Y.-H. and Plumlee, Matthew and Wild, Stefan M.},
  journal={Technometrics},
  volume={66},
  number={1},
  pages={1--13},
  year={2024},
  doi={10.1080/00401706.2023.2210170}
}

@article{rougier2008efficient,
  title={Efficient Emulators for Multivariate Deterministic Functions},
  author={Rougier, Jonathan},
  journal={Journal of Computational and Graphical Statistics},
  volume={17},
  number={4},
  pages={827--843},
  year={2008},
  doi={10.1198/106186008X384032}
}

@inproceedings{teh2005semiparametric,
  title={Semiparametric Latent Factor Models},
  author={Teh, Yee Whye and Seeger, Matthias and Jordan, Michael I.},
  booktitle={Proceedings of the Tenth International Workshop on Artificial Intelligence and Statistics},
  pages={333--340},
  year={2005},
  editor={Cowell, Robert G. and Ghahramani, Zoubin},
  volume={R5},
  series={Proceedings of Machine Learning Research},
  publisher={PMLR},
  url={https://proceedings.mlr.press/r5/teh05a.html}
}

@inproceedings{bonilla2008multi,
  title={Multi-Task Gaussian Process Prediction},
  author={Bonilla, Edwin V. and Chai, Kian M. and Williams, Christopher K. I.},
  booktitle={Advances in Neural Information Processing Systems 20},
  pages={153--160},
  year={2008},
  editor={Platt, J. C. and Koller, D. and Singer, Y. and Roweis, S.},
  publisher={Curran Associates, Inc.},
  url={https://proceedings.neurips.cc/paper/2007/hash/66368270ffd51418ec58bd793f2d9b1b-Abstract.html}
}

@article{alvarez2011computationally,
  title={Computationally Efficient Convolved Multiple Output Gaussian Processes},
  author={{\'A}lvarez, Mauricio A. and Lawrence, Neil D.},
  journal={Journal of Machine Learning Research},
  volume={12},
  number={41},
  pages={1459--1500},
  year={2011},
  url={https://www.jmlr.org/papers/v12/alvarez11a.html}
}

@article{alvarez2012kernels,
  title={Kernels for Vector-Valued Functions: A Review},
  author={{\'A}lvarez, Mauricio A. and Rosasco, Lorenzo and Lawrence, Neil D.},
  journal={Foundations and Trends in Machine Learning},
  volume={4},
  number={3},
  pages={195--266},
  year={2012},
  doi={10.1561/2200000036}
}

@inproceedings{bruinsma2020scalable,
  title={Scalable Exact Inference in Multi-Output Gaussian Processes},
  author={Bruinsma, Wessel and Perim, Eric and Tebbutt, William and Hosking, Scott and Solin, Arno and Turner, Richard},
  booktitle={Proceedings of the 37th International Conference on Machine Learning},
  pages={1190--1201},
  year={2020},
  editor={Daum{\'e} III, Hal and Singh, Aarti},
  volume={119},
  series={Proceedings of Machine Learning Research},
  publisher={PMLR},
  url={https://proceedings.mlr.press/v119/bruinsma20a.html}
}

@article{fricker2013multivariate,
  title={Multivariate Gaussian Process Emulators with Nonseparable Covariance Structures},
  author={Fricker, Thomas E. and Oakley, Jeremy E. and Urban, Nathan M.},
  journal={Technometrics},
  volume={55},
  number={1},
  pages={47--56},
  year={2013},
  doi={10.1080/00401706.2012.715835}
}

@article{paulo2012calibration,
  title={Calibration of Computer Models with Multivariate Output},
  author={Paulo, Rui and Garc{\'i}a-Donato, Gonzalo and Palomo, Jes{\'u}s},
  journal={Computational Statistics \& Data Analysis},
  volume={56},
  number={12},
  pages={3959--3974},
  year={2012},
  doi={10.1016/j.csda.2012.05.023}
}

@misc{surer2025puq,
  title={{PUQ} Users Manual},
  author={S{\"u}rer, {\"O}zge and O'Gara, David and Plumlee, Matthew and Wild, Stefan M.},
  number={Version 0.1.1},
  year={2025},
  url={https://github.com/parallelUQ/PUQ}
}

@article{liu2018remarks,
  title={Remarks on Multi-Output Gaussian Process Regression},
  author={Liu, Haitao and Cai, Jianfei and Ong, Yew-Soon},
  journal={Knowledge-Based Systems},
  volume={144},
  pages={102--121},
  year={2018},
  doi={10.1016/j.knosys.2017.12.034}
}

@article{gneiting2007strictly,
  title={Strictly Proper Scoring Rules, Prediction, and Estimation},
  author={Gneiting, Tilmann and Raftery, Adrian E.},
  journal={Journal of the American Statistical Association},
  volume={102},
  number={477},
  pages={359--378},
  year={2007},
  doi={10.1198/016214506000001437}
}

@article{li2026additive,
  title={Additive multi-index Gaussian process modeling, with application to multi-physics surrogate modeling of the quark-gluon plasma},
  author={Li, Kevin and Mak, Simon and Paquet, J-F and Bass, Steffen A},
  journal={Journal of the American Statistical Association},
  volume={121},
  number={553},
  pages={44--59},
  year={2026},
  publisher={Taylor \& Francis}
}

@misc{lartaud2022multi,
  title={Multi-output Gaussian processes for inverse uncertainty quantification in neutron noise analysis},
  author={Lartaud, Paul and Humbert, Philippe and Garnier, Josselin},
  journal={arXiv preprint arXiv:2211.02465},
  year={2022}
}

@article{plumlee2021high,
  title={High-fidelity hurricane surge forecasting using emulation and sequential experiments},
  author={Plumlee, Matthew and Asher, Taylor G and Chang, Won and Bilskie, Matthew V},
  journal={The Annals of Applied Statistics},
  volume={15},
  number={1},
  pages={460--480},
  year={2021},
  publisher={Institute of Mathematical Statistics}
}

@article{higdon2015bayesian,
  title={A Bayesian approach for parameter estimation and prediction using a computationally intensive model},
  author={Higdon, Dave and McDonnell, Jordan D and Schunck, Nicolas and Sarich, Jason and Wild, Stefan M},
  journal={Journal of Physics G: Nuclear and Particle Physics},
  volume={42},
  number={3},
  pages={034009},
  year={2015},
  publisher={IOP Publishing}
}

@article{hung2015analysis,
  title={Analysis of computer experiments with functional response},
  author={Hung, Ying and Joseph, V Roshan and Melkote, Shreyes N},
  journal={Technometrics},
  volume={57},
  number={1},
  pages={35--44},
  year={2015},
  publisher={Taylor \& Francis}
}

@article{plumlee2017bayesian,
  title={Bayesian calibration of inexact computer models},
  author={Plumlee, Matthew},
  journal={Journal of the American Statistical Association},
  volume={112},
  number={519},
  pages={1274--1285},
  year={2017},
  publisher={Taylor \& Francis}
}

@article{tuo2019adjustments,
  title={Adjustments to computer models via projected kernel calibration},
  author={Tuo, Rui},
  journal={SIAM/ASA Journal on Uncertainty Quantification},
  volume={7},
  number={2},
  pages={553--578},
  year={2019},
  publisher={SIAM}
}

@article{xie2021bayesian,
  title={Bayesian projected calibration of computer models},
  author={Xie, Fangzheng and Xu, Yanxun},
  journal={Journal of the American Statistical Association},
  volume={116},
  number={536},
  pages={1965--1982},
  year={2021},
  publisher={Taylor \& Francis}
}

@misc{ehlers2026varp,
  title={VarP-GP: cost-efficient Bayesian emulation of quark-gluon plasma modeling with variable statistical precision},
  author={Ehlers, R and Ji, Y and Jacobs, PM and Mak, S},
  journal={arXiv preprint arXiv:2603.05545},
  year={2026}
}

@article{mcnelis20183+,
  title={(3+ 1)-dimensional anisotropic fluid dynamics with a lattice QCD equation of state},
  author={McNelis, M and Bazow, D and Heinz, U},
  journal={Physical Review C},
  volume={97},
  number={5},
  pages={054912},
  year={2018},
  publisher={APS}
}

@article{mcnelis2021anisotropic,
title = {Anisotropic fluid dynamical simulations of heavy-ion collisions},
journal = {Computer Physics Communications},
volume = {267},
pages = {108077},
year = {2021},
issn = {0010-4655},
doi = {https://doi.org/10.1016/j.cpc.2021.108077},
url = {https://www.sciencedirect.com/science/article/pii/S0010465521001892},
author = {Mike McNelis and Dennis Bazow and Ulrich Heinz},
}

@misc{doob1953stochastic,
 author = {Doob, Joseph L.},
 title = {Stochastic processes},
 year = {1953},
 language = {English},
 howpublished = {New {York}: {Wiley}. 654 p. (1953).},
 zbMATH = {3085434},
 Zbl = {0053.26802}
}

\newpage

\appendix
\crefalias{section}{appendix}
\crefalias{subsection}{appendix}

\spacingset{1.5}

\begin{center}
    \Large Appendix for ``\titletext''
\end{center}

\section{Proof of \Cref{thm: ortho_kernel}}
\label{app:proof_ortho_kernel}

\begin{proof}
Fix $k$ and suppress the latent-component subscript. 
Because $c^*$ is continuous on $\mathcal{X} \times \mathcal{X},$ $z^*$ is mean-square continuous, implying continuity in probability.  Hence, a measurable $z^*$ via standard modification can be chosen \citep[ch.~2]{doob1953stochastic}.
Since $\mathcal{X}$ has finite measure and $\vect{g}$ and $c^*$ are
bounded, the stochastic integrals below exist and Fubini's theorem applies to
their first and second moments. Define
\[
\vect{A}:=\int_{\mathcal{X}}
\vect{g}(\vect{x})z^*(\vect{x})\,d\vect{x}.
\]
This random vector has mean zero. Using \Cref{eq:ck_star}, kernel symmetry, and
the definitions of $\vect{h}$ and $\mat{H}$ in \Cref{eq: h_ints},
\begin{align*}
\cov{\vect{A}}
&=
\int_{\mathcal{X}}\int_{\mathcal{X}}
\vect{g}(\vect{x})c^*(\vect{x},\vect{x}')
\vect{g}(\vect{x}')^\top\,d\vect{x}\,d\vect{x}'\\
&=
\mat{H}
-
\left(\int_{\mathcal{X}}\vect{g}(\vect{x})
\vect{h}(\vect{x})^\top\,d\vect{x}\right)
\mat{H}^{-1}
\left(\int_{\mathcal{X}}\vect{h}(\vect{x}')
\vect{g}(\vect{x}')^\top\,d\vect{x}'\right)\\
&=
\mat{H}-\mat{H}\mat{H}^{-1}\mat{H}
=\mat{0}_{r\times r}.
\end{align*}
Therefore,
$\E[\lVert\vect{A}\rVert_2^2]=\tr{\cov{\vect{A}}}=0$, so
$\vect{A}=\vect{0}$ with probability one. This proves part~(i).

Because $q$ is finite, part~(i) holds simultaneously for every latent
component. Hence,
\begin{align*}
\int_{\mathcal{X}}\vect{g}(\vect{x})
\vect{w}(\vect{x})^\top\,d\vect{x}
&=
\int_{\mathcal{X}}\vect{g}(\vect{x})
\vect{z}^*(\vect{x})^\top\,d\vect{x}\,\mat{\Psi}^\top\\
&=
\begin{bmatrix}
\vect{0} & \cdots & \vect{0}
\end{bmatrix}\mat{\Psi}^\top
=\mat{0}_{r\times p},
\end{align*}
which proves part~(ii).
\end{proof}

\section{Closed Forms for Squared-Exponential Integrals}
\label{app:analytic_integrals}

Following \citet{plumlee2018orthogonal}, we reproduce the component-wise closed forms
used by MOOGP. They assume a separable squared-exponential kernel on the
standardized domain $[-1,1]^d$ and basis functions formed from products of
distinct input coordinates. Output mixing and observation noise are applied only
after each $c_k^*$ is constructed, so neither enters the calculations below.

Suppress the latent-component subscript and let
\begin{equation*}
c(\vect{x},\vect{x}')
=\sigma^2\prod_{j=1}^d k_j(x_j,x_j'),
\qquad
k_j(u,v)=\exp\left\{-\frac{(u-v)^2}{\ell_j^2}\right\},
\label{eq:app_se_kernel}
\end{equation*}
where $\sigma^2>0$ and $\ell_j>0$. For
$g_s(\vect{x})=\prod_{j\in\mathcal{J}_s}x_j$, with
$\mathcal{J}_s\subseteq\{1,\ldots,d\}$ and the empty product equal to one,
define
\begin{align}
M_j(t)
&:=\int_{-1}^1 k_j(t,u)\,du, &
L_j(t)
&:=\int_{-1}^1 u\,k_j(t,u)\,du, \label{eq:app_ML_def}\\
I_{M,j}
&:=\int_{-1}^1\int_{-1}^1 k_j(u,v)\,du\,dv, &
I_{LL,j}
&:=\int_{-1}^1\int_{-1}^1
uv\,k_j(u,v)\,du\,dv. \label{eq:app_IMILL_def}
\end{align}
Their closed forms are
\begin{align}
M_j(t)
&=
\frac{\sqrt{\pi}\ell_j}{2}
\left[
\operatorname{erf}\left(\frac{1+t}{\ell_j}\right)
+
\operatorname{erf}\left(\frac{1-t}{\ell_j}\right)
\right], \label{eq:app_M_closed}\\
L_j(t)
&=
tM_j(t)
+
\frac{\ell_j^2}{2}
\left[
\exp\left\{-\frac{(1+t)^2}{\ell_j^2}\right\}
-
\exp\left\{-\frac{(1-t)^2}{\ell_j^2}\right\}
\right], \label{eq:app_L_closed}\\
I_{M,j}
&=
2\sqrt{\pi}\ell_j
\operatorname{erf}\left(\frac{2}{\ell_j}\right)
-
\ell_j^2\left(1-\exp\left\{-\frac{4}{\ell_j^2}\right\}\right), \label{eq:app_IM_closed}\\
I_{LL,j}
&=
\frac{\ell_j^4}{6}
\left(1-\exp\left\{-\frac{4}{\ell_j^2}\right\}\right)
-
\frac{\ell_j^2}{3}
\left(3-\exp\left\{-\frac{4}{\ell_j^2}\right\}\right)
+
\frac{2\sqrt{\pi}\ell_j}{3}
\operatorname{erf}\left(\frac{2}{\ell_j}\right).
\label{eq:app_ILL_closed}
\end{align}

Symmetry makes both mixed moments zero and $\mat{H}$ diagonal. Consequently,
\begin{equation*}
\begin{alignedat}{2}
h_s(\vect{x})
&={}\sigma^2
\prod_{j\in\mathcal{J}_s}L_j(x_j)
\prod_{j\notin\mathcal{J}_s}M_j(x_j),
&\qquad
H_{ss}
&={}\sigma^2
\prod_{j\in\mathcal{J}_s}I_{LL,j}
\prod_{j\notin\mathcal{J}_s}I_{M,j}.
\end{alignedat}
\label{eq:app_h_H_closed}
\end{equation*}
Thus, each orthogonalized latent covariance is
\begin{equation*}
c^*(\vect{x},\vect{x}')
=
\sigma^2\exp\left\{-\sum_{j=1}^d
\frac{(x_j-x_j')^2}{\ell_j^2}\right\}
-
\sum_{s=1}^r
\frac{h_s(\vect{x})h_s(\vect{x}')}{H_{ss}}.
\label{eq:app_cstar_closed}
\end{equation*}


\end{document}